\documentclass[conference,compsoc]{IEEEtran}
\ifCLASSOPTIONcompsoc
  \usepackage[nocompress]{cite}
\else
  \usepackage{cite}
\fi
\ifCLASSINFOpdf
\else
\fi
\usepackage{amsmath}
\usepackage{url}

\usepackage[available,functional,reproduced]{ieeebadges}

\usepackage{xcolor}
\usepackage{tabularx}
\usepackage{array}
\usepackage{multirow}
\usepackage{threeparttable}
\usepackage{xspace}
\usepackage{lipsum}
\usepackage{amsthm}
\usepackage{booktabs}
\usepackage{makecell}
\usepackage{amssymb}
\usepackage{graphicx}
\usepackage{hyperref}
\newcolumntype{C}{>{\centering\arraybackslash}X}
\newcommand{\1}{\mathbf{1}}                      
\newcommand{\amax}{\hat{x}_t}                      
\newtheorem{theorem}{Theorem}

\newtheorem{definition}{Definition}
\newtheorem{proposition}{Proposition}
\newtheorem{corollary}{Corollary}

\newcommand{\sys}{{\textsc{SecureSD}}\xspace}

\newenvironment{icompact}{
  \begin{list}{$\bullet$}{
    \itemindent -.05em
    \parsep 0pt plus 1pt
    \partopsep 0pt plus 1pt
    \topsep 2pt plus 2pt minus 2pt
    \itemsep 0pt plus 1.3pt
    \parskip 0pt plus 2pt
    \leftmargin 0.13in}
      }
{\normalsize
\end{list}
}

\usepackage[most]{tcolorbox}

\newcounter{takeaway}

\newenvironment{takeaway}{
  \refstepcounter{takeaway}
  \begin{tcolorbox}[
    colback=gray!10,
    colframe=gray!10,
    boxrule=0pt,
    arc=0pt,
    left=2pt,
    right=2pt,
    top=2pt,
    bottom=2pt,
    boxsep=0pt,
    breakable
  ]
  \noindent\textbf{Takeaway~\thetakeaway.}
}{
  \end{tcolorbox}
}
\usepackage{subcaption}
\usepackage[ruled,linesnumbered,vlined]{algorithm2e}

\makeatletter
\renewcommand\paragraph{%
  \@startsection{paragraph}{4}{\z@}%
  {1.15ex plus .2ex minus .2ex}
  {-0.3em}
  {\normalfont\normalsize\bfseries}}
\makeatother

\usepackage{todonotes}
\tcbuselibrary{raster}
\definecolor{orange}{RGB}{255,127,0}

\begin{document}
%
\title{Secure Speculative Decoding for Large Language Models}


\author{
\rm{Yichi Zhang$^1$,
Zhiqi Wang$^1$,
Neil Zhenqiang Gong$^2$,
Yuchen Yang$^1$}\\
$^1$The Pennsylvania State University,
$^2$Duke University \\
\{yichi.zhang, zhiqi.wang, yuchen.yang\}@psu.edu, neil.gong@duke.edu
}


%


\maketitle

\begin{abstract}

Speculative decoding accelerates inference for a large language model (LLM), referred to as the \emph{target model}, by first using a smaller model, referred to as the \emph{draft model}, to generate candidate tokens and then verifying them with the target model for acceptance or rejection. Prior studies primarily focused on the efficiency-utility trade-off of speculative decoding, e.g., lossy speculative decoding, leaving its security implications largely unexplored. 

In this work, we bridge this gap by providing the \emph{first} systematic study of the security implications of speculative decoding. Through a large-scale measurement study, we reveal a pronounced security-utility asymmetry: across a wide range of lossy speculative decoding methods, improvements in inference efficiency come at a disproportionately high cost to security, with attack success rates for jailbreak and prompt injection attacks increasing much faster than utility degrades.

We then propose \sys, a new theory-guided speculative decoding method that enhances security while maintaining efficiency and utility. Specifically, our theoretical analysis reveals that security degradation primarily originates from the early tokens generated by the draft model. Motivated by this insight, \sys applies a stricter verification criterion to draft-model tokens at early decoding positions. Extensive experiments on both security and utility benchmarks demonstrate that \sys significantly improves security while preserving efficiency and utility compared to existing speculative decoding methods.\footnote{The code and benchmarks for \sys are publicly available at \url{https://github.com/YichiCS/Secure-Speculative-Decoding}.}

\end{abstract}

%
\IEEEpeerreviewmaketitle

\section{Introduction}

As LLMs are increasingly deployed in interactive applications, autoregressive decoding has become a major latency bottleneck because each output token typically requires a separate forward pass through a large model. 
As shown in Algorithm~\ref{alg:specdec}, speculative decoding addresses this bottleneck by using a smaller \emph{draft model} to quickly generate candidate tokens for a larger LLM, referred to as the \emph{target model}; the target model then verifies these candidates in parallel and either accepts them or rejects them before continuing generation~\cite{leviathan2023fast,chen2023accelerating}.
By accepting multiple tokens in a single decoding step, speculative decoding can substantially reduce decoding latency. This technique has moved beyond research prototypes: Google reports applying speculative decoding in multiple products, including AI Overviews in Google Search~\cite{google_spec_decode_2024} and Gemma 4 models~\cite{google_gemma4_mtp}. Production-oriented serving stacks such as vLLM~\cite{vllm_spec_decode}, NVIDIA Triton/TensorRT-LLM~\cite{nvidia_triton_tensorrtllm_spec_decode}, Hugging Face Inference~\cite{hf_tgi_speculation}, and SGLang~\cite{sglang_spec_decode} also provide built-in speculative decoding support.

Existing work on speculative decoding has mainly focused on improving inference efficiency while preserving utility. 
Recent methods increase speedup by improving draft models~\cite{zhou2024distillspec, zafrir2025fastdraft, liu2024online, hu2026adaspec, li2026eagle}, using richer proposal structures~\cite{cai2024medusa, ankner2024hydra, miao2024specinfer, li2024eagle}, or adopting lossy verification strategies~\cite{bachmann2025judge, lu2025adasd, wang2025alignment, wang2026diversed} that accept more draft tokens per decoding step. 
Some methods further relax the verification criterion~\cite{kim2023speculative, narasimhan2025faster, holsman2025fuzzy, li2026training, song2026mars}, explicitly trading exact target-model behavior for higher throughput. 
However, these methods are still primarily evaluated through the efficiency--utility trade-off, using metrics such as speedup, average acceptance length, perplexity, task accuracy, or human preference. 
\emph{To the best of our knowledge, the safety and security implications of efficiency-oriented speculative decoding under adversarial inputs have not been systematically studied.} This gap matters because relaxed verification may accept draft tokens that the target model would otherwise reject. 
In safety settings, e.g., jailbreak~\cite{zou2023universal, jiang2024wildteaming, chao2024jailbreakbench, mazeika2024harmbench}, and security settings, e.g., prompt injection~\cite{liu2024formalizing}, such early deviations can steer the response from refusal toward compliance, causing safety and security risks that remain invisible to utility benchmarks.
For simplicity, we use \emph{security} to refer to both throughout the introduction. 

In this paper, we bridge this gap by providing the \emph{first} systematic study of the security implications of lossy speculative decoding. 
Across a range of speculative decoding methods, making token acceptance more aggressive improves inference efficiency with only modest degradation on standard utility benchmarks, but causes a disproportionately large increase in attack success rates on security benchmarks. 
In particular, jailbreak and prompt injection attacks become much more successful even when task accuracy, output quality, or other utility metrics remain relatively stable. 
This shows that security degradation is not simply another form of utility degradation: lossy speculative decoding can introduce small changes in the generation trajectory that are largely tolerated by general-purpose benchmarks but have much larger effects on security-critical behavior. 

To understand this asymmetry, we analyze how accepted draft tokens affect autoregressive generation. Our analysis shows that security degradation depends on two dimensions: (i) how far the accepted tokens deviate from the target model's original distribution, and (ii) at which token position these deviations occur in the output sequence. 
In particular, draft tokens accepted at early decoding positions have a disproportionate effect because they set the initial direction of the response. For adversarial prompts, an early deviation can shift generation from refusal toward compliance. This positional sensitivity explains why security can degrade much faster than utility.

Motivated by this insight, we propose \sys, a theory-guided and position-aware speculative decoding method that applies stricter verification to early draft tokens while retaining efficient verification later in the generation process. 
By correcting only the positions most responsible for security degradation, \sys avoids unnecessary verification overhead across the full response. Experiments on five utility and security benchmarks show that \sys reduces attack success rates of jailbreak and prompt injection attacks by up to 92.4\%, while retaining 99.8\% of the decoding speedup and keeping utility within 98.4\% of existing speculative decoding methods. Ablation studies further confirm that focusing verification on early positions is key to improving the security--utility--efficiency trade-off.

\paragraph{Contribution} We make the following contributions:
\begin{icompact}
    \item We provide the \emph{first} systematic study of the security implications of lossy speculative decoding, revealing a strong \textbf{security--utility asymmetry}: attack success rates rise much faster than utility degrades.

    \item We explain this asymmetry through theoretical and empirical analysis, showing that early draft tokens play an important role in security degradation.

    \item We propose \sys, a theory-guided, position-aware speculative decoding method that improves security while preserving efficiency and utility.
\end{icompact}

\section{Related Work}

\begin{algorithm}[t]
\caption{Vanilla Speculative Decoding}
\label{alg:specdec}
\KwIn{Draft model $p$, target model $q$, prefix $x_{1:t}$, speculation length $k$}
\KwOut{Tokens sampled from speculative decoding}
\For{$i = 1$ \KwTo $k$}{
    $p_i \gets p(\cdot \mid x_{1:t}, \tilde{x}_{1:i-1})$\;
    $\tilde{x}_i \sim p_i$\;
}
$\{q_i\}_{i=1}^{k+1} \gets \textsc{TargetForward}(x_{1:t}, \tilde{x}_{1:k})$\;

\For{$i = 1$ \KwTo $k$}{
    $r \sim \mathcal{U}(0,1)$\;
    \uIf{$r \le \min\{1,\, q_i(\tilde{x}_i)/p_i(\tilde{x}_i)\}$}{
        $x_{t+i} \gets \tilde{x}_i$\;
    }
    \Else{
        $x_{t+i} \sim
        \dfrac{(q_i-p_i)_{+}}
        {\sum_{x}(q_i(x)-p_i(x))_{+}}$\;
        \Return $x_{t+1},\ldots,x_{t+i}$\;
    }
}
$x_{t+k+1} \sim q_{k+1}$\;
\Return $x_{t+1},\ldots,x_{t+k+1}$\;
\end{algorithm}

\subsection{Speculative Decoding}
Speculative decoding~\cite{chen2023accelerating, leviathan2023fast} accelerates LLM generation by using a small draft model to propose candidate tokens and a larger target model to verify them in parallel, reducing expensive target-model computation while strictly preserving the target model's behavior. As shown in Algorithm~\ref{alg:specdec}, each decoding round proceeds in two phases: the draft model first autoregressively proposes $k$ candidate tokens, and the target model then computes their probabilities in a single parallel forward pass. Each candidate $\tilde{x}_i$ is accepted with probability $\min\{1,\, q_i(\tilde{x}_i)/p_i(\tilde{x}_i)\}$; upon the first rejection, a replacement token is sampled from the normalized residual distribution $(q_i-p_i)_{+}$ and the round terminates. If all $k$ candidates pass, one extra token is sampled from $q_{k+1}$, so each round commits between $1$ and $k+1$ tokens whose distribution provably matches sampling from the target model directly.

Recent works~\cite{kim2023speculative, narasimhan2025faster, holsman2025fuzzy, li2026training, song2026mars} further improve efficiency by relaxing the verification step. Instead of strictly preserving the target distribution, these methods accept more draft tokens when they appear sufficiently close to what the target model would have produced. This increases the acceptance rate and decoding speed with a small distributional shift from the original target model.

Specifically, LossySD~\cite{leviathan2023fast} relaxes the verification rule, allowing more draft tokens to pass verification under a controlled approximation. 
BiLD~\cite{kim2023speculative} uses model confidence to decide when to continue relying on the draft model and when to roll back to the target model. 
SpecCascade~\cite{narasimhan2025faster} formulates verification as a cascade decision, determining when the system can trust the draft model and when it should defer to the target model. 
FSD~\cite{holsman2025fuzzy} is designed to accept a draft token when the draft and target distributions are sufficiently similar. 
FLy~\cite{li2026training} uses uncertainty in the target distribution as a signal for relaxation, permitting more relaxation when the target model is confident enough.
MARS~\cite{song2026mars} relies on the structure of the target model's logits to determine whether a draft token is consistent with the target model's predictions. 

\emph{However, these works mainly focus on optimizing the efficiency--utility trade-off, leaving unclear how the distributional shifts introduced by relaxed verification affect security. We bridge this research gap in this paper.}

\subsection{Attacks on LLMs}

\paragraph{Jailbreak Attacks}
Jailbreak attacks are designed to guide LLMs to circumvent their safety rules and respond to queries that would otherwise be refused. Malicious users can craft natural-language prompts that reframe harmful requests through roleplay, fictional scenarios, or privilege escalation, such as ``Do-Anything-Now'' prompts~\cite{shen2024anything}, ``developer mode'' prompts, and character-roleplay prompts~\cite{liu2024hitchhiker}. More recent work formalizes these methods through persuasion~\cite{zeng2024johnny}, highlighting that LLMs are vulnerable to natural language prompting techniques.

Another line of jailbreak work relies on adaptive prompt-crafting algorithms~\cite{zou2023universal, liu2024autodan, anil2024many, chao2025jailbreaking} that use feedback from the target model. AutoDAN~\cite{liu2024autodan} optimizes human-crafted jailbreak prompts with a genetic algorithm by iteratively querying the target LLM and improving the prompt. Greedy Coordinate Gradient (GCG)~\cite{zou2023universal} optimizes an adversarial suffix to increase the likelihood of a compliant response from the target model, making it an iterative attack.

\paragraph{Prompt Injection Attacks}
Prompt injection attacks target the instruction-following behavior of LLM-integrated applications. In these attacks, an adversary supplies additional instructions that conflict with the intended task or higher-level system instructions~\cite{greshake2023not, perez2022ignore, liu2024formalizing, geng2026piarena}. Earlier work studies direct prompt injection~\cite{perez2022ignore}, where a user attempts to override the high-level instructions of an LLM-based application with commands such as ``ignore previous instructions.'' Later work shows that malicious instructions can also be planted in external sources, such as webpages, documents, or tool outputs, and then hijack a session between the benign user and LLM-based agents. This setting is known as indirect prompt injection~\cite{liu2024formalizing, geng2026piarena, greshake2023not}. Liu et al.~\cite{liu2024formalizing} formalize in-the-wild prompt injection techniques into a unified framework and show that combining existing techniques improves attack success rates. Adaptive or automatic jailbreak algorithms~\cite{zou2023universal, chao2025jailbreaking} can also be used to generate successful prompt injection prompts.

\textit{In this paper, we use jailbreak and prompt injection benchmarks to quantify the underexplored security--utility asymmetry introduced by relaxed speculative decoding: utility can remain largely stable while security degrades rapidly relative to the target model's security performance.}

\section{Threat Model}

\paragraph{System Model}
We consider a realistic deployment scenario in which the LLM infrastructure integrates speculative decoding with relaxed verification for further speedup. The system uses a large target model to generate responses and a cheaper draft model to propose candidate tokens. For the same input $X$, we compare the output of an LLM system with relaxed speculative decoding against that of an otherwise identical system with vanilla (i.e., lossless) speculative decoding. Our goal is to measure how relaxed verification changes the system's behavior relative to the vanilla baseline, and how much target-model performance it retains, treating the target model as the upper bound.

We focus on lossy speculative decoding systems that use relaxed verification to trade exact target-model behavior for throughput. Mainstream serving frameworks, including vLLM, NVIDIA TensorRT-LLM, and SGLang, support this setting~\cite{vllm_spec_decode,nvidia_triton_tensorrtllm_spec_decode,sglang_spec_decode}. We quantify relaxation by acceptance rate (AR) and evaluate its practical regime in Section~\ref{sec:measurement}.

\paragraph{Attacker's Objectives} 
We use adversarial objectives to define the safety and security stress-test setting. For a given query $X$, it produces an incorrect or unintended response $Y$. We primarily consider two attack scenarios. 
\begin{icompact}
    \item \textit{Jailbreak Attacks.} The attacker carefully crafts a query $X$ to induce the LLM system to generate responses $Y$ that violate human values or established safety policies.

    \item \textit{Prompt Injection Attacks.} The attacker carefully crafts a query $X$ to manipulate the LLM system such that it produces an unintended response $Y$ or causes a downstream application to execute potentially harmful actions $Y$.
\end{icompact}

\paragraph{Attacker's Capabilities}

We assume a black-box threat model in which the attacker can interact with the LLM system only through its public API. Specifically, the attacker can issue queries to an LLM system that employs relaxed speculative decoding and observe the generated responses. 

\paragraph{Attacker's Knowledge}

We assume that the attacker possesses partial knowledge of the target model, such as the target LLM's family or type (e.g., GPT-5.5). However, we do not assume knowledge of the exact relaxed speculative decoding configuration, including the draft model, verification rule, relaxation threshold, or decoding randomness. Since the same inputs are evaluated under the target model with and without relaxed speculative decoding, the measured security gap reflects the effect of relaxed verification rather than an exploitation of internal implementation details.

\begin{table*}[!t]
  \caption{\textbf{Average performance retention rate (PRR) and speedup (SPD) of lossless and six lossy speculative decoding methods across security and utility benchmarks.} For each speculative decoding method, we sweep its method-specific hyperparameters $\theta$ to obtain different ARs. We then report the average PRR and SPD as a function of AR. 
  We also report the PRR and SPD of lossless speculative decoding for comparison with the average performance of lossy speculative decoding.
  }
  \vspace{-0.1in}
  \label{tab:measurement_baselines}
  \begin{center}
    \begin{tabularx}{\textwidth}{rl*{8}{>{\centering\arraybackslash}X}}
      \toprule
      \multirow{4}{*}{Relaxed Verifier}
        & \multirow{4}{*}{Model Family}
        & \multicolumn{4}{c}{Safety/Security Benchmarks}
        & \multicolumn{4}{c}{Utility Benchmarks} \\
        \cmidrule(lr){3-6}\cmidrule(lr){7-10}
        &
        & \multicolumn{2}{c}{Jailbreak-SD}
        & \multicolumn{2}{c}{OpenPromptInjection}
        & \multicolumn{2}{c}{HumanEval}
        & \multicolumn{2}{c}{GSM8K} \\
      \cmidrule(lr){3-4}\cmidrule(lr){5-6}\cmidrule(lr){7-8}\cmidrule(lr){9-10}
        & & PRR$\uparrow$ & SPD$\uparrow$
        & PRR$\uparrow$ & SPD$\uparrow$
        & PRR$\uparrow$ & SPD$\uparrow$
        & PRR$\uparrow$ & SPD$\uparrow$ \\
      \midrule
      \multirow{2}{*}{LossySD~\cite{leviathan2023fast}}
        & Qwen3-0.6B/8B & .708 & 1.68 & .657 & 1.56 & .834 & 1.89 & .925 & 1.91 \\
        & Llama3-1B/8B  & .693 & 1.84 & .763 & 1.70 & .803 & 1.96 & .878 & 2.03 \\
      \midrule
      \multirow{2}{*}{BiLD~\cite{kim2023speculative}}
        & Qwen3-0.6B/8B & .674 & 1.50 & .664 & 1.35 & .856 & 1.63 & .945 & 1.74 \\
        & Llama3-1B/8B  & .568 & 1.59 & .718 & 1.46 & .748 & 1.73 & .903 & 1.71 \\
      \midrule
      \multirow{2}{*}{SpecCascade~\cite{narasimhan2025faster}}
        & Qwen3-0.6B/8B & .566 & 1.60 & .553 & 1.45 & .771 & 1.79 & .831 & 1.83 \\
        & Llama3-1B/8B  & .651 & 1.76 & .720 & 1.58 & .781 & 1.87 & .855 & 1.90 \\
      \midrule
      \multirow{2}{*}{FSD~\cite{holsman2025fuzzy}}
        & Qwen3-0.6B/8B & .927 & 1.60 & .759 & 1.56 & .865 & 1.58 & .940 & 1.59 \\
        & Llama3-1B/8B  & .610 & 1.85 & .734 & 1.76 & .850 & 1.64 & .945 & 1.50 \\
      \midrule
      \multirow{2}{*}{FLy~\cite{li2026training}}
        & Qwen3-0.6B/8B & .826 & 1.71 & .694 & 1.58 & .935 & 1.91 & .951 & 1.96 \\
        & Llama3-1B/8B  & .752 & 1.90 & .752 & 1.79 & .839 & 1.98 & .915 & 2.06 \\
      \midrule
      \multirow{2}{*}{MARS~\cite{song2026mars}}
        & Qwen3-0.6B/8B & .993 & 1.13 & .872 & 1.13 & .985 & 1.78 & .998 & 1.86 \\
        & Llama3-1B/8B  & .817 & 1.48 & .845 & 1.46 & .863 & 1.91 & .996 & 1.88 \\
      \midrule
      \multirow{2}{*}{\textbf{Average Lossy SD}}
        & Qwen3-0.6B/8B & .782 & 1.54 & .700 & 1.44 & .874 & 1.76 & .932 & 1.82 \\
        & Llama3-1B/8B  & .682 & 1.74 & .755 & 1.63 & .814 & 1.85 & .915 & 1.85 \\
      \midrule
      \multirow{2}{*}{\textbf{Lossless SD}}
        & Qwen3-0.6B/8B & 1.019 & 1.05 & .915 & 1.06 & .947 & 1.65 & 1.000 & 1.81 \\
        & Llama3-1B/8B  & .909 & 1.33 & .983 & 1.36 & .933 & 1.85 & .976 & 1.76 \\
      \bottomrule
    \end{tabularx}
  \end{center}
        \vspace{-0.15in}
\end{table*}

\section{Measurement: Security--Utility Asymmetry}
\label{sec:measurement}

We begin with a measurement study of how speculative decoding affects security performance as the token acceptance strategy becomes more aggressive. 

\subsection{Preliminaries on Speculative Decoding}

Speculative decoding methods share the same high-level draft-then-verify paradigm, but they differ in how the verifier decides whether to accept a proposed draft token and how generation proceeds after rejection. To compare these methods under a common framework, we first define a unified verifier abstraction.

\begin{definition}[Speculative Decoding Verifier]
\label{def:relaxed-sd-verifier}
A speculative decoding verifier can be represented as $S(q,p,\alpha,R)$, where the system operates over a sequence of draft positions and sequentially decides whether to accept or reject each draft token. At position $t$, conditioned on the current prefix $x_{<t}$, the target model and the draft model define the following conditional distributions:
\begin{align}
q_t(x) &= q(x_t = x \mid x_{<t}), \\
p_t(x) &= p(x_t = x \mid x_{<t}).
\end{align}
Here, $q_t$ specifies the target distribution that the decoding procedure aims to preserve, while $p_t$ specifies the draft distribution for efficient token generation.

Given a draft token $x$ at position $t$, the system applies an acceptance
function
\begin{align}
\alpha_t(x): \mathcal{V} \rightarrow [0,1],
\end{align}
which determines the probability with which the proposed token is accepted under the current prefix. If the proposed token is rejected, the system samples a replacement token from a recovery distribution
\begin{align}
R_t(x) &= R(x_t = x \mid x_{<t}).
\end{align}
Here, $R_t$ specifies the distribution used to recover from rejection and continue generation.
\end{definition}

This verifier definition can be broadly applied to current frameworks such as LossySD~\cite{leviathan2023fast}, BiLD~\cite{kim2023speculative}, SpecCascade~\cite{narasimhan2025faster}, FSD~\cite{holsman2025fuzzy}, FLy~\cite{li2026training}, and MARS~\cite{song2026mars}. 
We show how a speculative decoding method can be characterized by these four parameters in Table~\ref{tab:relaxed_verifiers}. Tuning the hyperparameters of this method allows us to modulate the aggressiveness of the corresponding token-acceptance strategies.

\begin{figure*}[t]
    \centering
    \includegraphics[width=\linewidth]{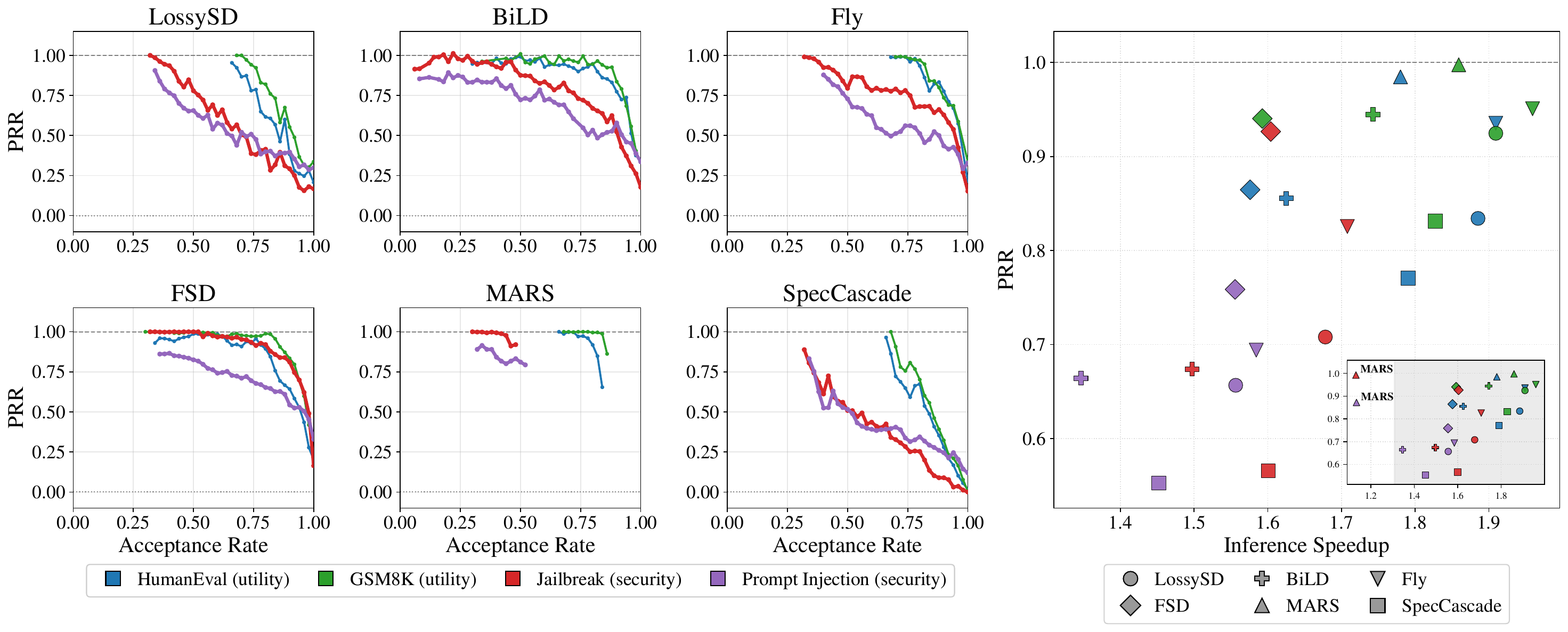}
    \caption{\textbf{Security--Utility Asymmetry under Speculative Decoding.} The left panel shows the PRR-AR curves of six speculative decoding methods on four benchmarks. The right panel shows the trade-off between PRR and SPD. The separation between the curves and scatter points for security and utility tasks reveals the security--utility asymmetry.}
    \label{fig:security_utility_asymmetry_qwen}
      \vspace{-0.10in}
\end{figure*}

\subsection{Experimental Setup}
\label{sec:measurement_es}

\paragraph{Benchmarks}
We evaluate the system on utility, safety, and security benchmarks.
For utility evaluation, we use two benchmarks commonly adopted in speculative decoding studies: 164 Python programming problems from HumanEval~\cite{chen2021evaluating} to assess code generation, and 200 problems sampled from GSM8K~\cite{cobbe2021training} to assess mathematical reasoning.
We consider jailbreak for safety settings and prompt injection for security settings. For jailbreak evaluation, we construct a benchmark named Jailbreak-SD of 200 examples from four widely used datasets: WildJailbreak~\cite{jiang2024wildteaming}, JailbreakBench~\cite{chao2024jailbreakbench}, HarmBench~\cite{mazeika2024harmbench}, and AdvBench~\cite{zou2023universal}. 
For prompt injection, we use OpenPromptInjection~\cite{liu2024formalizing} to evaluate general prompt injection attacks. We sample two instructions from different sources across seven task datasets and combine them using five attack strategies, resulting in a benchmark of 200 samples, with 40 samples for each attack strategy.
We also conduct a full evaluation on AgentDojo~\cite{debenedetti2024agentdojo} to assess more realistic prompt injection scenarios in long-context tasks.

\begin{table}[t]
\centering
\setlength{\tabcolsep}{3pt}
\caption{\textbf{Instantiation of relaxed verifiers.}
We instantiate six speculative decoding methods with relaxed verification under Definition~\ref{def:relaxed-sd-verifier}.}
\label{tab:relaxed_verifiers}
\begin{threeparttable}
\begin{tabularx}{\columnwidth}{@{}r c c C c@{}}
\toprule
\textbf{Method} & \textbf{$q_t$} & \textbf{$p_t$} & \textbf{$\alpha_t(x)$} & \textbf{$R_t$} \\
\midrule
LossySD~\cite{leviathan2023fast}
  & $q_t$ & $p_t$
  & $\min\{1, q_t(x)/p_t(x)+\epsilon\}$
  & $\rho_t(q_t;x)$\tnote{1} \\
\midrule
BiLD~\cite{kim2023speculative}
  & $q_t$ & $p_t$
  & \makecell[c]{$\1[\max_v p_t(v)\ge \mathrm{fb}]$\\
    $\cdot \1[q_t(x)\ge e^{-\mathrm{rb}}\lor x=\amax]$}
  & $\delta_{\amax}$ \\
\midrule
SpecCascade~\cite{narasimhan2025faster}
  & $\pi_t^r$ & $p_t$
  & $\min\{1, \pi_t^r(x)/p_t(x)\}$
  & $\rho_t(\pi_t^r;x)$ \\
\midrule
FSD~\cite{holsman2025fuzzy}
  & $q_t$ & $p_t$
  & \makecell[c]{$\1[D_r(q_t\|p_t)\le\tau]$}
  & $q_t$ \\
\midrule
FLy~\cite{li2026training}
  & $q_t$ & $p_t$
  & \makecell[c]{$\1[x=\amax]$\\
    $\lor\1[\tilde H(q_t)\ge\theta \wedge \neg \mathrm{win}^{W}_t]$}
  & $q_t$ \\
\midrule
MARS~\cite{song2026mars}
  & $q_t$ & $p_t$
  & \makecell[c]{$\1[x=\amax]$\\
    $\lor\1[x=\mathrm{top\text{-}2}(q_t)\wedge z_2/z_1>\tau]$}
  & $\delta_{\amax}$ \\
\bottomrule
\end{tabularx}
\begin{tablenotes}[flushleft] \footnotesize \item[1] $\rho_t(\mu;x)= (\mu(x)-p_t(x))_+/\sum_v(\mu(v)-p_t(v))_+$ denotes the positive-residual recovery distribution, where $(f)_+=\max\{f,0\}$. \end{tablenotes}
\end{threeparttable}
      \vspace{-0.15in}
\end{table}

\paragraph{Metrics}
We evaluate speculative decoding using three metrics that capture its efficiency, utility, and security.
\begin{icompact}
    \item \emph{Acceptance Rate (AR).}
    AR measures acceptance aggressiveness as the fraction of draft tokens accepted. We sweep the method-specific hyperparameters in Table~\ref{tab:relaxed_verifiers} to obtain AR values from conservative verification to aggressive acceptance, forming an approximately continuous scale between 0 and 1.

    \item \emph{Performance Retention Rate (PRR).}
    To compare target-model performance retention across tasks with different score scales, we orient all \emph{raw scores} so that \textbf{higher is better}: Pass@1 for HumanEval, Accuracy for GSM8K, attack failure rate (AFR) for Jailbreak-SD and OpenPromptInjection, and utility under attack (UUA) for AgentDojo.
    We use the official evaluators to compute Pass@1, Accuracy, and UUA. For AFR, we follow the StrongREJECT~\cite{souly2024strongreject} setting and use GPT-4o-mini as an LLM-as-a-judge~\cite{gu2026survey} to determine whether the attack succeeds.
    Let $x$ be the raw score of a certain speculative decoding method, and let $x_{\text{target}}$ and $x_{\text{draft}}$ be the raw scores of the target and draft models under standard decoding. We define
    \begin{align}
    \mathrm{PRR} = \frac{x - x_{\text{draft}}}{x_{\text{target}} - x_{\text{draft}}},
    \end{align}
    which measures the fraction of the target performance retained by speculative decoding: $\mathrm{PRR}=0$ matches the draft model, and $\mathrm{PRR}=1$ matches the target model. We also report raw scores for direct interpretation in Table~\ref{tab:raw_model_performance}.

    \item \emph{Inference Speedup (SPD).}
    SPD is the ratio between the wall-clock decoding time of the target model under standard decoding and that of speculative decoding. 
    Larger values indicate faster inference, e.g., $\mathrm{SPD}=2$ means that speculative decoding achieves a $2\times$ speedup against standard autoregressive decoding of the target model.

\end{icompact}

\paragraph{Speculative Decoding Methods}
We consider the lossless method and six lossy speculative decoding methods: LossySD~\cite{leviathan2023fast}, BiLD~\cite{kim2023speculative}, FLy~\cite{li2026training}, FSD~\cite{holsman2025fuzzy}, MARS~\cite{song2026mars}, and SpecCascade~\cite{narasimhan2025faster}. 
We use a default speculation length of $k=5$, meaning that the draft model proposes five candidate tokens at each decoding step. Unless otherwise specified, we use a generation batch size of 32 and set the decoding temperature to 0, corresponding to greedy decoding. These are widely adopted default settings.
For notational simplicity, we use $\theta$ to abstractly denote the tunable hyperparameters of each speculative decoding method.

\paragraph{Models}
We evaluate four draft-target model pairs. We pair Qwen3-0.6B (i.e., draft model) with Qwen3-8B and Qwen3-32B as target models, and Llama-3.2-1B-Instruct (i.e., draft model) with Llama-3.1-8B-Instruct and Llama-3.3-70B-Instruct as target models.
We report the standard evaluation metrics of these models on the four benchmarks in Table~\ref{tab:raw_model_performance}, and use these results to compute PRR in the subsequent sections.

\begin{table}[t]
\centering
\footnotesize
\setlength{\tabcolsep}{8pt}
\caption{\textbf{Raw Model Performance.} We report the task performance of four models across five benchmarks, highlighting the performance gap between the draft and target models. We use these results to compute PRR.}
\begin{tabular}{lcccc}
\toprule
\multirow{3}{*}{Task (Metric)} & \multicolumn{2}{c}{Qwen3} & \multicolumn{2}{c}{Llama3} \\
\cmidrule(lr){2-3} \cmidrule(lr){4-5}
& 0.6B & 8B & 1B & 8B \\
\midrule
Jailbreak-SD (AFR$\uparrow$) & .195 & .895 & .105 & .915 \\
OpenPromptInjection (AFR$\uparrow$) & .030 & .820 & .060 & .830 \\
AgentDojo (UUA$\uparrow$) & .078 & .264 & -- & -- \\
HumanEval (Pass@1$\uparrow$) & .396 & .854 & .427 & .701 \\
GSM8K (Accuracy$\uparrow$) & .500 & .930 & .410 & .820 \\
\bottomrule
\end{tabular}
\label{tab:raw_model_performance}
\vspace{-0.15in}
\end{table}

\paragraph{Environment} 
All experiments are conducted on a server equipped with an AMD EPYC 9334 32-Core Processor running Ubuntu 22.04.5 LTS, with four NVIDIA RTX PRO 6000 Blackwell Workstation Edition GPUs.

\subsection{Security Degrades Faster than Utility}

\paragraph{Lower Security Performance Retention} 
We first investigate how lossy speculative decoding affects security and utility benchmarks across a broad range of verifier settings. Table~\ref{tab:measurement_baselines} reports the PRR and SPD of different lossy speculative decoding methods, averaged over hyperparameter sweeps that span diverse ARs. We have two observations:

\begin{icompact}
    \item \emph{Lossy speculative decoding retains lower PRR on security benchmarks than on utility benchmarks.} For Qwen3-0.6B/8B, the average PRR is 0.744 on the two security benchmarks, Jailbreak and Prompt Injection, compared with 0.905 on the two utility benchmarks, HumanEval and GSM8K. The same trend holds for Llama3-1B/8B, with average PRRs of 0.719 on security benchmarks and 0.866 on utility benchmarks.

    \item \emph{The lower PRR on security benchmarks does not come with a higher SPD.} For Qwen3-0.6B/8B, lossy speculative decoding achieves only 1.49$\times$ average SPD on security benchmarks, compared with 1.79$\times$ on utility benchmarks. For Llama3-1B/8B, the average SPD is 1.69$\times$ on security benchmarks and 1.85$\times$ on utility benchmarks. Therefore, security benchmarks suffer larger performance degradation, but do not receive higher SPD in return.
\end{icompact}

\paragraph{Earlier Security Degradation} 
We next measure how fast security and utility performance degrade as the verifier accepts more draft tokens, i.e., with a higher AR. Figure~\ref{fig:security_utility_asymmetry_qwen} shows how the PRR varies with the AR on utility and security benchmarks. Moving right along a curve corresponds to a more aggressive acceptance strategy with higher SPD, but the generated sequence may deviate further from the target model's original behavior. We have two observations:

\begin{icompact}
    \item \emph{Security PRR is lower at comparable ARs.} 
    PRR decreases as the AR increases on both utility and security benchmarks, indicating that a more aggressive relaxed verifier introduces greater behavioral deviation from the target model. However, this degradation is consistently larger for security tasks. Under the same verifier and at comparable ARs, security benchmarks retain substantially less target model performance than utility benchmarks.

    \item \emph{Security PRR drops much earlier.} 
    Security degradation also begins at lower ARs. For several methods, utility PRR on GSM8K and HumanEval remains close to the baseline until the verifier becomes highly aggressive, e.g., at ARs around 0.7. In contrast, security PRR often starts to decline at much lower ARs, e.g., around 0.3.
\end{icompact}

These observations show that speculative decoding with relaxed verifiers can weaken jailbreak or prompt injection defenses even when performance on general utility benchmarks shows only limited degradation. 

\paragraph{Practical Operating Regime}
Compared with lossless SD, lossy SD raises average SPD from $1.05\times$ to $1.49\times$ on Qwen and from $1.35\times$ to $1.69\times$ on Llama. Combined with the thresholds above, $\mathrm{AR}\in[0.3,0.7]$ identifies the regime where additional speedup is meaningful and safety/security degradation has emerged.

\begin{takeaway}
Lossy speculative decoding exhibits a security--utility asymmetry: at comparable ARs, it can preserve utility while incurring earlier and sharper degradation on security. These results suggest that efficiency--utility evaluation alone is insufficient for assessing speculative decoding in security-sensitive deployments.
\end{takeaway}
\section{\sys}

We now propose \sys, a verification correction mechanism that can be seamlessly integrated into any speculative decoding pipeline. \sys mitigates the performance degradation introduced by relaxed speculative decoding verifiers, particularly on security-critical tasks. 

\subsection{\sys Overview} 

\begin{algorithm}[t]
\caption{\sys}
\label{alg:sys}
\KwIn{Relaxed verifier $S$ with acceptance functions $\{\alpha_i\}_{i=1}^{k}$ and recovery distributions $\{R_i\}_{i=1}^{k+1}$, prefix $x_{1:t}$, speculation length $k$, verification correction schedule $\{\eta_j\}_{j\ge 1}$}
\KwOut{Tokens sampled under \sys}

\For{$i = 1$ \KwTo $k$}{\label{line:begin-draft}
    $p_i \gets p(\cdot \mid x_{1:t}, \tilde{x}_{1:i-1})$ \label{line:draft-dist}\;
    $\tilde{x}_i \sim p_i$ \label{line:draft-sample}\;
}
$\{q_i\}_{i=1}^{k+1} \gets \textsc{TargetForward}(x_{1:t}, \tilde{x}_{1:k})$ \label{line:target}\;

\For{$i = 1$ \KwTo $k$}{
    $\alpha_i^{\mathrm{sd}} \gets \min\{\mathbf{1},\, q_i/p_i\}$ \label{line:sd-alpha}\;
    $R_i^{\mathrm{sd}} \gets
    \dfrac{(q_i-p_i)_{+}}
    {\sum_{x}(q_i(x)-p_i(x))_{+}}$ \label{line:sd-recovery}\;
    $\alpha_i^{\eta} \gets
    (1-\eta_{t+i})\alpha_i^{\mathrm{sd}}
    + \eta_{t+i}\alpha_i$ \label{line:mix-alpha}\;
    $r \sim \mathcal{U}(0,1)$ \label{line:uniform}\;

    \uIf{$r \le \alpha_i^{\eta}(\tilde{x}_i)$}{
        $x_{t+i} \gets \tilde{x}_i$ \label{line:accept}\;
    }
    \Else{
        $R_i^{\eta} \gets
        (1-\eta_{t+i})R_i^{\mathrm{sd}}
        + \eta_{t+i}R_i$ \label{line:mix-recovery-reject}\;
        $x_{t+i} \sim R_i^{\eta}$ \label{line:recovery-sample-reject}\;
        \Return $x_{t+1},\ldots,x_{t+i}$ \label{line:return-reject}\;
    }
}

$R_{k+1}^{\eta} \gets
(1-\eta_{t+k+1})q_{k+1}
+ \eta_{t+k+1}R_{k+1}$ \label{line:mix-recovery-accept}\;
$x_{t+k+1} \sim R_{k+1}^{\eta}$ \label{line:extra-sample}\;
\Return $x_{t+1},\ldots,x_{t+k+1}$ \label{line:return-accept}\;
\end{algorithm}

Algorithm~\ref{alg:sys} presents \sys, a speculative decoding algorithm that incorporates verification correction.
The algorithm is formulated as a single-step procedure, and the complete decoding process can be constructed by invoking it repeatedly.
The input consists of a relaxed verifier $S$, the current prefix $x_{1:t}$, a speculation length $k$, and a verification correction schedule $\{\eta_j\}_{j\ge 1}$. The verifier $S$ exposes acceptance functions $\{\alpha_i\}_{i=1}^{k}$ and recovery distributions $\{R_i\}_{i=1}^{k+1}$. 
The key idea of \sys is to use a correction schedule to augment the relaxed verifier with a lossless reference path induced by the standard verifier. 
We provide both theoretical and empirical analysis of this design in Section~\ref{sec:analysis}.

\paragraph{Draft Generation and Scoring (Lines~\ref{line:begin-draft}--\ref{line:target})}
Given the prefix $x_{1:t}$, the draft model first generates $k$ speculative tokens autoregressively. At speculative position $i$, it forms the draft distribution
\begin{align}
    p_i = p(\cdot \mid x_{1:t}, \tilde{x}_{1:i-1}),
\end{align}
and samples $\tilde{x}_i\sim p_i$ (Lines~\ref{line:begin-draft}--\ref{line:draft-sample}).
These tokens are tentative and are not committed to the output sequence until they pass verification. The target model is then applied in a single forward pass to the concatenation of the current prefix and the speculative draft tokens, yielding target distributions $\{q_i\}_{i=1}^{k+1}$ for the $k$ speculative positions and the subsequent extra-token position (Line~\ref{line:target}). This single target forward pass preserves the parallelization benefit of speculative decoding.

\paragraph{Lossless Reference of the Standard Verifier (Lines~\ref{line:sd-alpha}--\ref{line:sd-recovery})}
For each token position, \sys constructs the lossless reference of the standard verifier~\cite{leviathan2023fast}. The acceptance rule is
\begin{align}
    \alpha_i^{\mathrm{sd}}
    =
    \min\{\mathbf{1}, q_i/p_i\},
    \label{eq:sd-alpha}
\end{align}
and the corresponding recovery distribution is
\begin{align}
    R_i^{\mathrm{sd}}
    =
    \frac{(q_i-p_i)_{+}}
    {\sum_x (q_i(x)-p_i(x))_{+}}.
    \label{eq:sd-recovery}
\end{align}
Both of these ensure that speculative decoding remains distributionally lossless.

\paragraph{Verification Correction (Lines~\ref{line:mix-alpha}--\ref{line:return-accept})}
\sys interpolates the standard and relaxed verifiers at each token position. The mixed acceptance probability is defined as
\begin{align}
    \alpha_i^{\eta}
    =
    (1-\eta_{t+i})\alpha_i^{\mathrm{sd}}
    +
    \eta_{t+i}\alpha_i.
    \label{eq:mixed-alpha}
\end{align}
The proposed draft token $\tilde{x}_i$ is accepted if a uniform random variable $r\sim\mathcal{U}(0,1)$ is no larger than $\alpha_i^{\eta}(\tilde{x}_i)$ (Lines~\ref{line:uniform}--\ref{line:accept}). Upon acceptance, the algorithm sets $x_{t+i}\gets \tilde{x}_i$. If the token is rejected, verification stops immediately. The algorithm then constructs a mixed recovery distribution at the rejected position:
\begin{align}
    R_i^{\eta} = (1-\eta_{t+i})R_i^{\mathrm{sd}} + \eta_{t+i}R_i,
    \label{eq:mixed-recovery-reject}
\end{align}
samples $x_{t+i}\sim R_i^{\eta}$, and returns the accepted prefix together with this recovery token, as shown in Lines~\ref{line:mix-recovery-reject}--\ref{line:return-reject}. This left-to-right stopping rule preserves the causal structure of autoregressive decoding, since all draft tokens after a rejection depend on a token that has not been accepted.

If all $k$ draft tokens are accepted, \sys constructs the sampling distribution for the next position by interpolating between the target distribution $q_{k+1}$ and the relaxed verifier's recovery distribution $R_{k+1}$:
\begin{align}
R_{k+1}^{\eta} = (1-\eta_{t+k+1})q_{k+1} + \eta_{t+k+1}R_{k+1}.
\label{eq:mixed-recovery-accept}
\end{align}
It then samples the extra token $x_{t+k+1}\sim R_{k+1}^{\eta}$ and returns $x_{t+1},\ldots,x_{t+k+1}$ (Lines~\ref{line:mix-recovery-accept}--\ref{line:return-accept}). Therefore, each invocation of \sys advances generation by at least one token and by at most $k+1$ tokens.

\subsection{Verification Correction Schedule $\{\eta_j\}_{j\ge 1}$}
\label{sec:vsc}
The verification correction schedule $\{\eta_j\}_{j\ge 1}$ in Algorithm~\ref{alg:sys} determines how \sys interpolates between the standard verifier and the relaxed verifier. The index $j$ denotes the \emph{global} decoding position. 
When $\eta_j=0$, the verification procedure reduces to standard speculative decoding. When $\eta_j=1$, \sys recovers the relaxed verifier $S$ and directly applies its acceptance and recovery rules. For $0<\eta_j<1$, \sys interpolates between the two verification rules, allowing it to balance distributional fidelity and decoding efficiency at each token position.

Thus, $\{\eta_j\}_{j\ge 1}$ determines the trajectory of the correction level. We define three correction schedules $\{\eta_j\}_{j\ge 1}$, motivated by the empirical analysis in Section~\ref{subsec:empirical_analysis}.

\medskip\noindent\emph{(i) Step schedule.}
The step schedule keeps the correction level at its initial value before position $L$ and switches directly to the final value at position $L$:
\begin{align}
\eta_{\mathrm{step}}(j) = \1[j> L] = 
    \begin{cases}
        0, & j\le L,\\
        1, & j> L.
    \end{cases}
\end{align}
This schedule makes an abrupt transition between the two boundary correction levels at the switch position $L$.

\medskip\noindent\emph{(ii) Linear schedule.}
The linear schedule interpolates between the two boundary correction levels at a constant rate until the terminal value is reached:
\begin{align}
\eta_{\mathrm{linear}}(j)=\mathrm{clip}\left(\frac{j}{L+1},0,1\right).
\end{align}
This schedule changes $\eta_j$ linearly, with larger $L$ producing a slower transition and smaller $L$ producing a faster one.

\medskip\noindent\textit{(iii) Power-law schedule.}
The power-law schedule generalizes the linear transition with a curvature parameter $\gamma>0$:
\begin{align}
\eta_{\gamma}(j) = \left[\mathrm{clip}\left(\frac{j}{L+1},0,1\right)\right]^{\gamma}.
\end{align}
For $\gamma=1$, the power-law schedule reduces to the linear schedule. Values $\gamma>1$ delay the transition toward later positions, while values $0<\gamma<1$ make it occur earlier.

\begin{takeaway}
\sys uses the verification correction schedule to control how strongly each decoding position follows the standard verifier or the relaxed verifier. This behavior is determined by the correction schedule $\{\eta_j\}_{j\ge 1}$, which allows \sys to balance utility and efficiency without changing the verifier itself.
\end{takeaway}
\section{Why \sys Works: Analysis}
\label{sec:analysis}

This section provides theoretical and empirical analyses of how \sys mitigates performance degradation under lossy speculative decoding. The theoretical analysis characterizes the task-agnostic performance-retention benefit of verifier correction, while the empirical analysis identifies the greater early-position sensitivity of security metrics. This distinction explains why \sys applies stronger correction to early tokens on security tasks.

\subsection{Theoretical Analysis of \sys}
\label{sec:theo}

\paragraph{Output Distribution}
A speculative decoding step at position $t$ proceeds as follows. Conditioned on the prefix $x_{<t}$, the draft model first proposes a token $\tilde{x}_t \sim p_t$. The proposed token is then accepted with probability $\alpha_t(\tilde{x}_t)$. If it is accepted, the system sets $x_t = \tilde{x}_t$; otherwise, the system samples a replacement token from the recovery distribution $x_t \sim R_t$. Equivalently, the one-step transition of the system can be written as
\begin{align}
x_t =
\begin{cases}
\tilde{x}_t, & \text{with probability } \alpha_t(\tilde{x}_t), \\
\hat{x}_t, \ \hat{x}_t \sim R_t, & \text{with probability } 1-\alpha_t(\tilde{x}_t).
\end{cases}
\end{align}
Through this definition, we can derive the output distribution of speculative decoding with the relaxed verifier.

\begin{theorem}[Output Distribution of Relaxed Verifiers]
\label{thm:one_step_output_distribution}
Fix a prefix $x_{<t}$ and consider a relaxed speculative decoding verifier $S(q,p,\alpha,R)$. Let $\tilde{q}_t$ denote the output distribution induced by one decoding step of $S$, i.e.,
\begin{align}
\tilde{q}_t(x) \triangleq \Pr_S(x_t = x \mid x_{<t}).
\end{align}
Then, for every token $x \in \mathcal{V}$, $\tilde{q}_t(x)$ is given by
\begin{align}
\tilde{q}_t(x) = p_t(x)\alpha_t(x) + \left(1-\mathbb{E}_{u\sim p_t}[\alpha_t(u)]\right)R_t(x).
\label{eq:induced_output_distribution}
\end{align}
The proof is provided in Appendix~\ref{app:proof_one_step_output_distribution}.
\end{theorem}

\begin{corollary}[Output Distribution of Standard SD]
\label{cor:lossless_sd}
For a fixed prefix $x_{<t}$, consider the standard SD verifier with acceptance function and recovery distribution defined as
\begin{align}
\alpha_t^{\mathrm{sd}}(x) &= \min\left(1,\frac{q_t(x)}{p_t(x)}\right), \\
R_t^{\mathrm{sd}}(x) &= \frac{(q_t(x)-p_t(x))_+}{\sum_{u\in\mathcal{V}}(q_t(u)-p_t(u))_+},
\end{align}
where $(f)_+ = \max(0,f)$. Then the one-step output distribution induced by
the verifier matches the target distribution:
\begin{align}
\tilde{q}_t^{\mathrm{sd}}(x) = q_t(x),
\qquad \forall x\in\mathcal{V}.
\end{align}
The proof is provided in Appendix~\ref{app:proof_lossless_sd}.
\end{corollary}

Corollary~\ref{cor:lossless_sd} establishes that the standard verifier is a zero-deviation baseline for the subsequent analysis.

\paragraph{Distributional Deviation}

We then measure the one-step distortion of a relaxed verifier by the total variation distance $\mathrm{TV}(\tilde{q}_t, q_t)$ between its output $\tilde{q}_t$ and the target $q_t$. The following theorem decomposes $\mathrm{TV}(\tilde{q}_t, q_t)$ into two sources corresponding to the two design choices of the verifier.

\begin{theorem}[Decomposition of Distributional Deviation]
\label{thm:distributional_deviation_decomposition}
For any relaxed verifier $S(q,p,\alpha,R)$ at draft position $t$, define the pointwise acceptance deviation $\delta_t(x) \triangleq \alpha_t(x)-\alpha_t^{\mathrm{sd}}(x)$. It follows that 
\begin{align}
\mathrm{TV}(\tilde{q}_t, q_t) &\le
\underbrace{\mathbb{E}_{x\sim p_t}\left[\left|\delta_t(x)\right|\right]
}_{\text{acceptance mismatch}}\notag\\
&\quad+\underbrace{\left(1-\mathbb{E}_{x\sim p_t}[\alpha_t(x)]\right)\mathrm{TV}(R_t,R_t^{\mathrm{sd}})}_{\text{recovery mismatch}}.
\label{eq:distributional_deviation_decomposition}
\end{align}
The proof is provided in Appendix~\ref{app:proof_distributional_deviation_decomposition}.
\end{theorem}

\begin{takeaway}
Theorem~\ref{thm:distributional_deviation_decomposition} attributes the distributional deviation to two factors: (1) the acceptance bias introduced by the relaxed verifier, and (2) the discrepancy between the recovery distribution and the standard recovery distribution. 
This observation suggests that \sys requires a coordinated design for correcting \emph{both} $\alpha$ and $R$.
\end{takeaway}

\paragraph{Performance Degradation Decomposition}

We connect the distributional deviation $\mathrm{TV}(\tilde{q}_t, q_t)$ to downstream task metrics. Let
$S_m:\mathcal{V}^{T}\rightarrow[0,1]$ be the scoring function for the specific task $m$, and let
\begin{align}
M = S_m(x_{1:T}) \in [0,1]
\end{align}
denote the task score. Write $\tilde{Q}$ (resp.\ $Q$) for the full-sequence distribution induced by the relaxed verifier (resp.\ standard verifier). For a fixed prefix $x_{<t}$, define 
\begin{align}
\mu_t(x;x_{<t}) \triangleq \mathbb{E}_{Q} \left[M \mid x_{<t},x_t=x\right],
\end{align}
as the expected task score when token $x$ is chosen at position $t$ and the remaining tokens are drawn from $q$.

\begin{theorem}[Task Metric Drift Bound]
\label{thm:bd_bound}
For any task metric $M\in[0,1]$, let $\tilde{Q}_{<t}$ be the marginal of $\tilde{Q}$ over $x_{<t}$, write $\mathrm{osc}(f)\triangleq\max_x f(x)-\min_x f(x)$, and write $\mathrm{TV}_t(x_{<t})$ for the distributional deviation $\mathrm{TV}(\tilde{q}_t,q_t)$ of Theorem~\ref{thm:distributional_deviation_decomposition}, viewed as a function of the prefix $x_{<t}$. Define
\begin{align}
\Gamma \triangleq \sum_{t=1}^{T}\mathbb{E}_{x_{<t}\sim\tilde{Q}_{<t}}\bigl[
\mathrm{TV}_t(x_{<t})\cdot\mathrm{osc}\left(\mu_t(\cdot;x_{<t})\right)\bigr].
\label{eq:Gamma_def}
\end{align}
With $\Gamma$ as defined in Equation~\eqref{eq:Gamma_def}, the metric gap satisfies
\begin{align}
\left|\mathbb{E}_{\tilde{Q}}[M]-\mathbb{E}_{Q}[M]\right|\le\min\{\Gamma,1\}.
\label{eq:bd_bound_prefix}
\end{align}
Letting $\mathrm{TV}_t\triangleq\sup_{x_{<t}}\mathrm{TV}_t(x_{<t})$ denote the worst-case distributional deviation and $\mathrm{osc}(\mu_t)\triangleq\sup_{x_{<t}}\mathrm{osc}(\mu_t(\cdot;x_{<t}))$ denote the worst-case task-score oscillation, we have
\begin{align}
\left|\mathbb{E}_{\tilde{Q}}[M]-\mathbb{E}_{Q}[M]\right|
\le\min\left\{\sum_{t=1}^{T}\mathrm{TV}_t\cdot\mathrm{osc}(\mu_t),1\right\}.
\label{eq:bd_bound_uniform}
\end{align}
The proof is provided in Appendix~\ref{app:proof_bd_bound}.
\end{theorem}

Theorem~\ref{thm:bd_bound} links performance degradation to the distributional deviation of the relaxed verifier. Notably, the bound depends not only on the deviation between $\tilde{q}_t$ and $q_t$, but also on the task sensitivity to the token chosen at position $t$. 
This sensitivity is captured by $\mathrm{osc}(\mu_t(\cdot;x_{<t}))$, which measures the largest possible change in the final metric induced by a one-token decision. We formalize it as follows.

\begin{definition}[Positional Task Sensitivity]
\label{def:task_criticality}
The positional task sensitivity at position $t$ given prefix $x_{<t}$ is
\begin{align}
I_t(x_{<t}) &\triangleq \mathrm{osc}\left(\mu_t(\cdot;x_{<t})\right) \notag\\
&= \max_{x\in\mathcal{V}}\mu_t(x;x_{<t})
- \min_{x\in\mathcal{V}}\mu_t(x;x_{<t}).
\label{eq:task_criticality_def}
\end{align}
\end{definition}

\begin{takeaway}
Theorem~\ref{thm:bd_bound} shows that performance degradation introduced by the relaxed verifier can be decomposed into: (1) the distributional deviation $\mathrm{TV}(\tilde{q}_t,q_t)$, and (2) the task's sensitivity $I_t$ to token choices at each position.
\end{takeaway}

\paragraph{Distributional Deviation of \sys}

We now analyze how \sys reduces the distributional deviation.
Recall from Algorithm~\ref{alg:sys} that at position $t$, \sys applies the mixed acceptance function and recovery distribution
\begin{align}
\alpha_t^{\eta} &= (1-\eta_t)\alpha_t^{\mathrm{sd}}+\eta_t\alpha_t,
\label{eq:mixed_alpha_analysis}\\
R_t^{\eta} &= (1-\eta_t)R_t^{\mathrm{sd}}+\eta_t R_t.
\label{eq:mixed_recovery_analysis}
\end{align}
Applying Theorem~\ref{thm:one_step_output_distribution} to the mixed verifier $S(q,p,\alpha^{\eta},R^{\eta})$ gives the output distribution:
\begin{align}
\tilde{q}_t^{\eta}(x) = p_t(x)\alpha_t^{\eta}(x) + \left(1-\bar\alpha_t^{\eta}\right)R_t^{\eta}(x),
\label{eq:eta_mixed_output}
\end{align}
where $\bar\alpha_t^{\eta}\triangleq\mathbb{E}_{u\sim p_t}[\alpha_t^{\eta}(u)]=(1-\eta_t)\bar\alpha_t^{\mathrm{sd}}+\eta_t\bar\alpha_t$. 

\begin{proposition}[Degradation Mitigation via \sys]
\label{prop:eta_mixing_deviation}
Under $\eta_t\in[0,1]$ at position $t$, the distributional deviation of the \sys verifier satisfies
\begin{align}
\mathrm{TV}(\tilde{q}_t^{\eta},q_t) \leq &\eta_t\mathrm{TV}(\tilde{q}_t,q_t) \notag\\
+& \eta_t(1-\eta_t)|\bar{\alpha}_t-\bar{\alpha}_t^{\mathrm{sd}}|\mathrm{TV}(R_t,R_t^{\mathrm{sd}}).
\label{eq:eta_mixing_bound}
\end{align}
Let $\epsilon_t = \eta_t(1-\eta_t)|\bar{\alpha}_t-\bar{\alpha}_t^{\mathrm{sd}}|\mathrm{TV}(R_t,R_t^{\mathrm{sd}})$. The task metric drift under \sys satisfies
\begin{align}
\left|\mathbb{E}_{\tilde{Q}^{\eta}}[M]\!-\!\mathbb{E}_{Q}[M]\right|\!\le\!\min\left\{\sum_{t=1}^{T}(\eta_t\mathrm{TV}_t+\epsilon_t)I_t,1\right\},
\label{eq:eta_bd_bound}
\end{align}
where $\mathrm{TV}_t$, $\epsilon_t$, and $I_t$ are uniform upper bounds over prefixes. The proof is provided in Appendix~\ref{app:proof_eta_mixing_deviation}.
\end{proposition}

Proposition~\ref{prop:eta_mixing_deviation} demonstrates that the performance degradation introduced by a relaxed verifier can be mitigated by reducing $\eta_t \mathrm{TV}_t + \epsilon_t$. Importantly, we find that $\epsilon_t$ is typically zero, and even when non-zero, it is negligible compared to the dominant term $\eta_t \mathrm{TV}_t$. For instance, under the STEP schedule, $\epsilon_t = 0$ because $\eta_t \in \{0, 1\}$; similarly, in LossySD, $\epsilon_t = 0$ because $R_t = R_t^{\mathrm{sd}}$. For cases where $\epsilon_t \neq 0$, we evaluate its magnitude relative to $\eta_t \mathrm{TV}_t$ in Table~\ref{tab:epsilon_magnitude}. The ratio of $\epsilon_t$ to $\eta_t \mathrm{TV}_t$ ranges from $0\%$ to $2.55\%$, averaging only $0.52\%$. This indicates that the bound is dominated by $\eta_t \mathrm{TV}_t$ and thus scales approximately linearly with $\eta_t$. We observe a consistent empirical trend in Figure~\ref{fig:eta_mitigation}.

\begin{table}[t]
    \centering
    \caption{
    \textbf{Magnitude of $\epsilon_t$ relative to $\eta_t\mathrm{TV}_t$.} Across four datasets, $\epsilon_t$ is negligible compared to $\eta_t\mathrm{TV}_t$.
    }
        \setlength{\tabcolsep}{14pt}
    \label{tab:epsilon_magnitude}
    \begin{tabular}{lccc}
        \toprule
        Method & $\eta=0.25$ & $\eta=0.5$ & $\eta=0.75$ \\
        \midrule
        LossySD      & 0.00\% & 0.00\% & 0.00\% \\
        FSD          & 1.03\% & 0.67\% & 0.33\% \\
        MARS         & 0.48\% & 0.32\% & 0.16\% \\
        BiLD         & 0.62\% & 0.41\% & 0.21\% \\
        FLy          & 2.55\% & 1.70\% & 0.85\% \\
        SpecCascade  & 0.00\% & 0.00\% & 0.00\% \\
        \bottomrule
    \end{tabular}
    \vspace{-0.15in}
\end{table}

\begin{takeaway}
Theorem~\ref{thm:bd_bound} relates the performance degradation under lossy speculative decoding to both the output distribution shift induced by the relaxed verifier and the task type. Both Proposition~\ref{prop:eta_mixing_deviation} and our empirical analysis demonstrate that \sys can reduce $\eta_t \mathrm{TV}_t$ through correction to mitigate performance degradation.
\end{takeaway}

\subsection{Empirical Study of Theoretical Analysis}
\label{subsec:empirical_analysis}

The analysis in Theorem~\ref{thm:bd_bound} attributes the performance degradation of relaxed SD to two factors. The first is the per-step distributional deviation $\mathrm{TV}(\tilde{q}_t, q_t)$ characterized in Theorem~\ref{thm:distributional_deviation_decomposition}, and the second is the positional task sensitivity $I_t$ defined in Definition~\ref{def:task_criticality}. It further predicts that $\eta$-mixing linearly contracts the deviation, as shown in Proposition~\ref{prop:eta_mixing_deviation}. We now empirically validate these three predictions.

\begin{figure}[t]
    \centering
    \includegraphics[width=\linewidth]{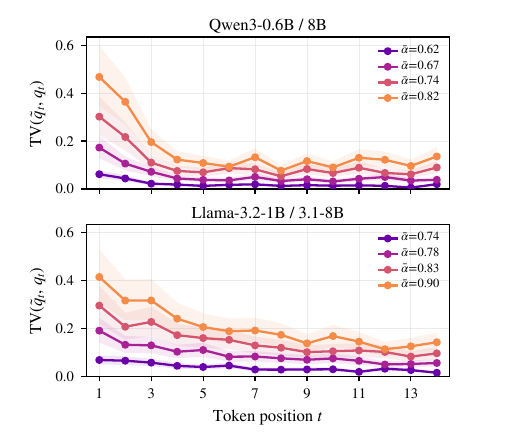}
    \caption{\textbf{Distributional Deviation of Relaxed Verifiers.} Per-position deviation $\mathrm{TV}(\tilde{q}_t, q_t)$ under relaxed verifiers concentrates at early token positions and is amplified by higher acceptance rates.}
    \label{fig:tv_deviation}
      \vspace{-0.1in}
\end{figure}

\paragraph{Distributional Deviation of Relaxed Verifiers}
We estimate the per-position deviation $\mathrm{TV}(\tilde{q}_t, q_t)$ on both security and utility benchmarks in Figure~\ref{fig:tv_deviation}.
The deviation shows a clear positional pattern: it is concentrated at early generation positions and rapidly diminishes as decoding proceeds. As the acceptance rate increases, the deviation becomes larger, suggesting that more permissive verifiers induce stronger departures from the target distribution. This amplification is most evident at early positions, where deviations can propagate through the autoregressive decoding process.

\begin{figure}[t]
    \centering
    \includegraphics[width=\linewidth]{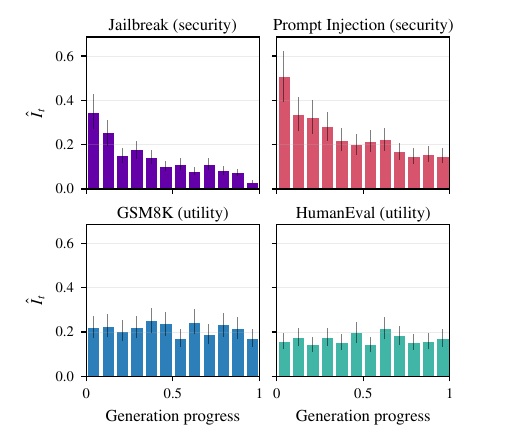}
    \caption{\textbf{Positional Task Sensitivity.} $\hat{I}_t$ measures the task sensitivity to tokens at different positions. Security benchmarks show early concentrated $\hat{I}_t$, while utility benchmarks show a more uniform distribution across generation.}
    \label{fig:task_sensitivity}
          \vspace{-0.15in}
\end{figure}

\paragraph{Positional Task Sensitivity}
Computing the task sensitivity $I_t$ exactly is intractable, since it requires evaluating $\mu_t(x;x_{<t})$ for every token $x\in\mathcal{V}$. We therefore estimate it with Monte Carlo sampling. For each sampled position $t$, we select the top-$k$ candidate tokens under $q_t$, generate $N$ target-model continuations for each candidate, and score the resulting completions with the task metric. We then use the range of the average scores as the estimate $\hat{I}_t$.

Figure~\ref{fig:task_sensitivity} shows how the positional task sensitivity $I_t$ changes over the course of generation. We find that security and utility benchmarks exhibit different $I_t$ patterns. For security benchmarks, $I_t$ is mostly concentrated in the early generated tokens, whereas for utility benchmarks, $I_t$ is more evenly distributed across the generation process. 
This difference arises because utility-oriented tasks, such as math and coding, often require step-by-step reasoning, where correctness depends on maintaining accuracy throughout the entire generation. In contrast, performance on security benchmarks is more strongly determined by the model's early response behavior, such as whether it gives a safe or unsafe initial response and whether the generation is hijacked toward unintended content.

\begin{takeaway}
Proposition~\ref{prop:eta_mixing_deviation} and the following empirical analysis suggest that the performance degradation is mainly controlled by $\eta_t \cdot \mathrm{TV}_t \cdot I_t$. On security tasks, both the distributional deviation and task sensitivity are concentrated in the first few tokens. This system-level result for draft-target model pairs aligns with prior findings on early-token sensitivity in single-model safety~\cite{qi2025safety,zhang2024large}. Therefore, using smaller $\eta_t$ at early positions can substantially mitigate the degradation. This observation motivates \sys in Algorithm~\ref{alg:sys}.
\end{takeaway}

\paragraph{Degradation Mitigation via \sys}

We first validate the relationship between $\eta$ and the distributional deviation predicted by Proposition~\ref{prop:eta_mixing_deviation}. Figure~\ref{fig:eta_mitigation} (left) plots the measured deviation $\mathrm{TV}(\tilde{q}_t^{\eta}, q_t)$ as $\eta$ varies from $1$ to $0$. The deviation decreases nearly linearly with $\eta$, closely matching the contraction bound in Proposition~\ref{prop:eta_mixing_deviation}.

We next examine whether applying $\eta$-verification correction to the first few tokens can recover model performance.
Figure~\ref{fig:eta_mitigation} (right) studies a representative setting where utility PRR remains high but security PRR drops substantially. Under the step correction schedule, security PRR increases with $L$ and gradually saturates.
This trend matches the positional pattern of $\mathrm{TV}_t \cdot I_t$ on security tasks and supports our theoretical analysis.
Large $L$ values, however, gradually reduce SPD by introducing more strict verification steps. In practice, \sys should choose the correction schedule according to the desired PRR--SPD trade-off.

\begin{figure}[t]
    \centering
    \includegraphics[width=\linewidth]{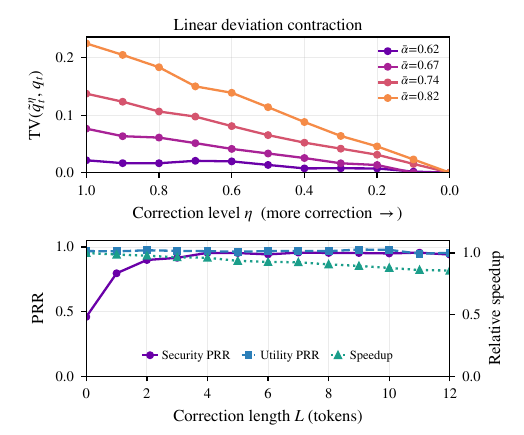}
    \caption{\textbf{Degradation Mitigation via \sys.} \emph{Left:} the measured deviation $\mathrm{TV}(\tilde{q}_t^{\eta}, q_t)$ contracts linearly in $\eta$. 
    \emph{Right:} security PRR shows a saturating increase as $L$ grows. }
    \label{fig:eta_mitigation}
          \vspace{-0.15in}
\end{figure}

\begin{takeaway}
Both theoretical and empirical analyses support the framework in three key aspects: relaxed verification induces distributional deviation, security tasks are highly sensitive to early-token errors, and $\eta$-mixing contracts deviation as predicted. These findings justify \sys: enforcing stronger early correction for security while relaxing later verification for efficiency.
\end{takeaway}

\begin{table*}[!t]
  \caption{
  \textbf{Impact of \sys on LLM System Safety/Security/Utility/Efficiency.}
  We evaluate \sys (default) with six relaxed/lossy verifiers on Qwen3-0.6B/8B and Llama3-1B/8B across four datasets. The average $\Delta_\mathrm{PRR}$ and $\Delta_\mathrm{SPD}$ show that \sys improves security and safety while largely maintaining target-model utility and lossy SD efficiency.
  }
  \vspace{-0.1in}
  \label{tab:main_results}
  \begin{center}
    \setlength{\tabcolsep}{4pt}
    \begin{tabularx}{\textwidth}{rl*{8}{>{\centering\arraybackslash}X}}
      \toprule
      \multirow{4}{*}{Relaxed Verifier}
        & \multirow{4}{*}{Model Family}
        & \multicolumn{4}{c}{Safety/Security Benchmarks}
        & \multicolumn{4}{c}{Utility Benchmarks} \\
        \cmidrule(lr){3-6}\cmidrule(lr){7-10}
        &
        & \multicolumn{2}{c}{Jailbreak-SD}
        & \multicolumn{2}{c}{OpenPromptInjection}
        & \multicolumn{2}{c}{HumanEval}
        & \multicolumn{2}{c}{GSM8K} \\
      \cmidrule(lr){3-4}\cmidrule(lr){5-6}\cmidrule(lr){7-8}\cmidrule(lr){9-10}
        & & $\Delta_\mathrm{PRR}$ & $\Delta_{\mathrm{SPD}}$
        & $\Delta_\mathrm{PRR}$ & $\Delta_{\mathrm{SPD}}$
        & $\Delta_\mathrm{PRR}$ & $\Delta_{\mathrm{SPD}}$
        & $\Delta_\mathrm{PRR}$ & $\Delta_{\mathrm{SPD}}$ \\
      \midrule
      \multirow{2}{*}{LossySD~\cite{leviathan2023fast}}
        & Qwen3-0.6B/8B & +.337 & +.050 & +.125 & -.010 & -.013 & +.001 & +.025 & -.020 \\
        & Llama3-1B/8B  & +.045 & +.003 & +.071 & -.002 & +.025 & -.002 & +.010 & -.005 \\
      \midrule
      \multirow{2}{*}{BiLD~\cite{kim2023speculative}}
        & Qwen3-0.6B/8B & +.295 & +.063 & +.086 & -.010 & -.054 & -.016 & +.026 & -.010 \\
        & Llama3-1B/8B  & +.039 & +.018 & +.073 & -.005 & -.002 & +.001 &  .000 & -.004 \\
      \midrule
      \multirow{2}{*}{SpecCascade~\cite{narasimhan2025faster}}
        & Qwen3-0.6B/8B & +.598 & +.083 & +.229 & -.016 & +.016 & -.005 & +.196 & +.006 \\
        & Llama3-1B/8B  & +.074 & +.037 & +.197 & +.001 & +.022 & +.004 & +.091 & +.008 \\
      \midrule
      \multirow{2}{*}{FSD~\cite{holsman2025fuzzy}}
        & Qwen3-0.6B/8B & +.122 & -.023 & +.084 & -.016 & -.007 & -.004 & +.007 & -.004 \\
        & Llama3-1B/8B  & +.051 & +.020 & +.122 & +.003 & +.008 & -.001 & +.002 & -.002 \\
      \midrule
      \multirow{2}{*}{FLy~\cite{li2026training}}
        & Qwen3-0.6B/8B & +.194 & -.038 & +.152 & -.018 & -.022 & -.006 & +.003 & -.005 \\
        & Llama3-1B/8B  & +.007 & +.002 & +.095 & -.007 & +.035 & -.001 & +.004 & -.003 \\
      \midrule
      \multirow{2}{*}{MARS~\cite{song2026mars}}
        & Qwen3-0.6B/8B & +.072 & +.021 & +.009 & -.014 & +.034 & -.017 & +.018 & -.012 \\
        & Llama3-1B/8B  & +.036 & +.032 & +.038 & +.001 & +.018 & +.002 & -.014 & -.003 \\
      \midrule
      \multirow{2}{*}{\textbf{Average Lossy SD}}
        & Qwen3-0.6B/8B & +.270 & +.026 & +.114 & -.014 & -.008 & -.008 & +.046 & -.007 \\
        & Llama3-1B/8B  & +.042 & +.019 & +.099 & -.001 & +.018 & +.001 & +.015 & -.001 \\
      \bottomrule
    \end{tabularx}
  \end{center}
    \vspace{-0.1in}
\end{table*}

\begin{figure*}[t]
    \centering
    \includegraphics[width=\linewidth]{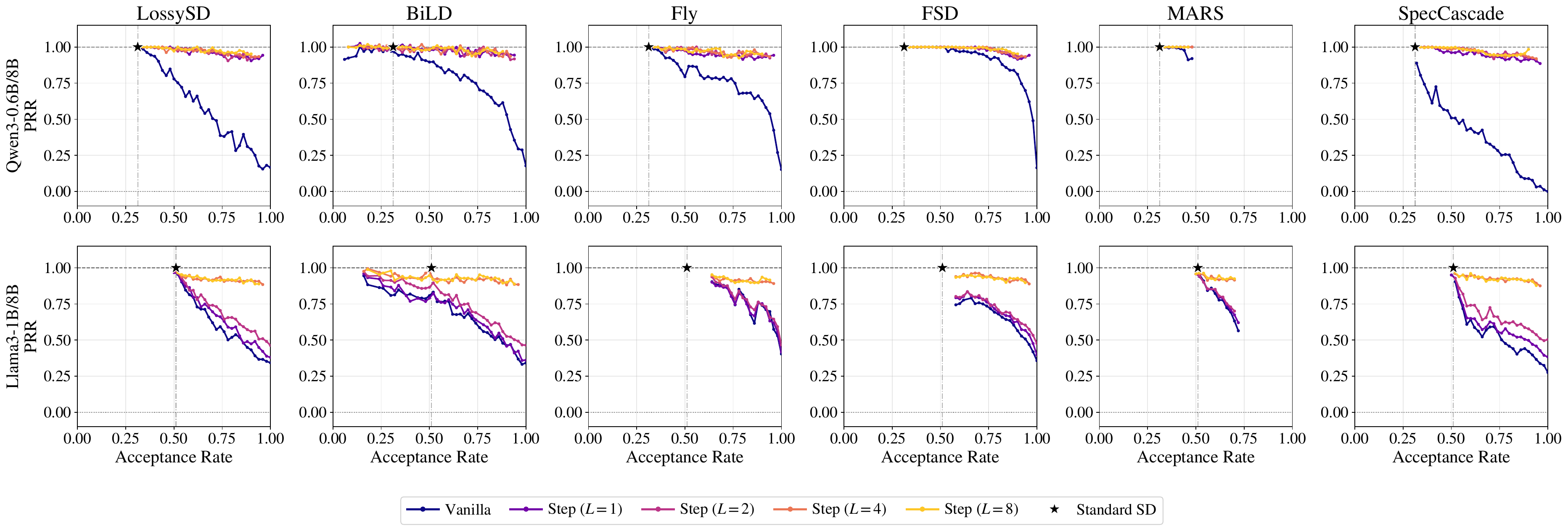}
    \caption{\textbf{Security PRR versus acceptance rate on Jailbreak-SD.} Columns are relaxed verifiers and rows are model families (Qwen3-0.6B/8B, Llama3-1B/8B). The star marks lossless speculative decoding. As the AR increases, vanilla PRR degrades sharply, while \sys (STEP~($L$)) keeps it near the lossless level.}
    \label{fig:step_l_curves_jailbreak}
      \vspace{-0.1in}
\end{figure*}

\section{Experiments}

In this section, we conduct a comprehensive evaluation of \sys to assess its ability to improve security while preserving utility and efficiency. We also design extensive ablation studies to demonstrate that \sys can be applied to LLM systems with different parameter scales while remaining effective.

\subsection{Experimental Setup}

We follow the same experimental setup as in Section~\ref{sec:measurement_es} and introduce two additional metrics to quantify the performance gains of \sys. 
We use the step schedule with $L=1$ as the default setting of \sys.

\paragraph{Performance Retention Rate Gain $\Delta_\mathrm{PRR}$}
$\Delta_\mathrm{PRR}$ measures how much additional target-model performance \sys retains on top of a relaxed verifier. To isolate the effect of \sys, we compare in a \emph{paired} manner: the two runs use the same verifier with the same hyperparameter $\theta$, and differ only in whether \sys is applied. Let $\mathrm{PRR}^{\theta}_{\text{\sys}}$ and $\mathrm{PRR}^{\theta}_{\text{vanilla}}$ denote the PRRs of the verifier with and without \sys under hyperparameter $\theta$. We define
\begin{align}
\Delta_\mathrm{PRR} = \frac{1}{|\Theta|}\sum_{\theta \in \Theta}
\left( \mathrm{PRR}^{\theta}_{\text{\sys}} - \mathrm{PRR}^{\theta}_{\text{vanilla}} \right),
\end{align}
where $\Theta$ is the hyperparameter set of the verifier. A positive $\Delta_\mathrm{PRR}$ indicates that \sys retains more target model performance than the vanilla verifier.

\paragraph{Inference Speedup Gain $\Delta_{\mathrm{SPD}}$}
$\Delta_{\mathrm{SPD}}$ quantifies the efficiency benefit of applying \sys under the same paired comparison. Let $\mathrm{SPD}^{\theta}_{\text{\sys}}$ and $\mathrm{SPD}^{\theta}_{\text{vanilla}}$ denote the inference speedups (SPD) of the verifier with and without \sys under the same hyperparameter $\theta$. We define
\begin{align}
\Delta_{\mathrm{SPD}} = \frac{1}{|\Theta|}\sum_{\theta \in \Theta}
\left( \mathrm{SPD}^{\theta}_{\text{\sys}} - \mathrm{SPD}^{\theta}_{\text{vanilla}} \right).
\end{align}
A small negative $\Delta_{\mathrm{SPD}}$ indicates that \sys introduces a minor efficiency overhead, while a positive value indicates that \sys is faster than the vanilla verifier.
\begin{table*}[!t]
  \caption{
  \textbf{Impact of the correction length $L$ on \sys.} 
  $\Delta_\mathrm{PRR}$ increases consistently with $L$, indicating performance gains from correcting the generation of more tokens, especially on Llama3-1B/8B. $\Delta_\mathrm{SPD}$ changes only slightly as $L$ increases, indicating that this behavior incurs only a minor efficiency overhead. Results are averaged across six lossy-SD methods. 
  }
  \vspace{-0.1in}
  \label{tab:eta_length}
  \begin{center}
    \begin{tabularx}{\textwidth}{c*{16}{>{\centering\arraybackslash}X}}
      \toprule
      \multirow{4}{*}{$L$}
        & \multicolumn{8}{c}{Qwen3-0.6B/8B}
        & \multicolumn{8}{c}{Llama3-1B/8B} \\
      \cmidrule(lr){2-9}\cmidrule(lr){10-17}
        & \multicolumn{2}{c}{Jailbreak-SD}
        & \multicolumn{2}{c}{Prompt Inj.}
        & \multicolumn{2}{c}{HumanEval}
        & \multicolumn{2}{c}{GSM8K}
        & \multicolumn{2}{c}{Jailbreak-SD}
        & \multicolumn{2}{c}{Prompt Inj.}
        & \multicolumn{2}{c}{HumanEval}
        & \multicolumn{2}{c}{GSM8K} \\
      \cmidrule(lr){2-3}\cmidrule(lr){4-5}\cmidrule(lr){6-7}\cmidrule(lr){8-9}
      \cmidrule(lr){10-11}\cmidrule(lr){12-13}\cmidrule(lr){14-15}\cmidrule(lr){16-17}
        & $\Delta_\mathrm{PRR}$ & $\Delta_\mathrm{SPD}$
        & $\Delta_\mathrm{PRR}$ & $\Delta_\mathrm{SPD}$
        & $\Delta_\mathrm{PRR}$ & $\Delta_\mathrm{SPD}$
        & $\Delta_\mathrm{PRR}$ & $\Delta_\mathrm{SPD}$
        & $\Delta_\mathrm{PRR}$ & $\Delta_\mathrm{SPD}$
        & $\Delta_\mathrm{PRR}$ & $\Delta_\mathrm{SPD}$
        & $\Delta_\mathrm{PRR}$ & $\Delta_\mathrm{SPD}$
        & $\Delta_\mathrm{PRR}$ & $\Delta_\mathrm{SPD}$ \\
      \midrule
        1 & +.270 & +.026 & +.114 & -.014 & -.008 & -.008 & +.046 & -.007 & +.042 & +.019 & +.099 & -.001 & +.018 & +.001 & +.015 & -.001 \\
        2 & +.279 & +.006 & +.152 & -.021 & -.009 & -.009 & +.045 & -.009 & +.096 & +.014 & +.109 & -.011 & +.023 & +.000 & +.021 & -.011 \\
        4 & +.282 &  .000 & +.213 & -.026 & -.008 & -.015 & +.032 & -.018 & +.321 & -.008 & +.140 & -.022 & +.023 & -.001 & +.027 & -.005 \\
        8 & +.285 & -.047 & +.224 & -.035 & -.005 & -.023 & +.034 & -.035 & +.325 & -.052 & +.135 & -.035 & +.027 & -.010 & +.033 & -.019 \\
      \bottomrule
    \end{tabularx}
  \end{center}
    \vspace{-0.15in}
\end{table*}
\begin{table}[t]
\centering
\caption{
\textbf{Evaluation of \sys on AgentDojo.}
We evaluate \sys with Qwen3-0.6B/8B on the realistic, long-context prompt injection benchmark AgentDojo. The \textit{utility under attack} of the target and draft models is .264 and .078, respectively. \sys improves the average PRR while preserving SPD.
}
\label{tab:agentdojo_results}
\setlength{\tabcolsep}{4pt}
\begin{tabular}{lcccccc}
\toprule
\multirow{3}{*}{Method} & \multicolumn{2}{c}{w/o \sys} & \multicolumn{2}{c}{w/ \sys} & \multirow{3}{*}{$\Delta_{\mathrm{PRR}}\uparrow$} & \multirow{3}{*}{$\Delta_{\mathrm{SPD}}\uparrow$} \\
\cmidrule(lr){2-3} \cmidrule(lr){4-5}
& PRR$\uparrow$ & SPD$\uparrow$ & PRR$\uparrow$ & SPD$\uparrow$ & & \\
\midrule
LossySD     & .548 & 1.919 & .723 & 1.947 & +.175 & +.028 \\
BiLD        & .932 & 1.872 & .983 & 1.891 & +.051 & +.019 \\
SpecCascade & .169 & 2.081 & .328 & 2.158 & +.159 & +.077 \\
FSD         & .869 & 1.839 & .955 & 1.847 & +.086 & +.007 \\
FLy         & .793 & 1.798 & .904 & 1.799 & +.111 & +.001 \\
MARS        & .863 & 1.882 & .953 & 1.910 & +.090 & +.028 \\
\midrule
\textbf{Average} & .696 & 1.898 & .808 & 1.925 & +.112 & +.027 \\
\bottomrule
\end{tabular}
  \vspace{-0.15in}
\end{table}

\subsection{Main Results}

Table~\ref{tab:main_results} summarizes the paired improvements of \sys, measured by $\Delta_\mathrm{PRR}$ and $\Delta_{\mathrm{SPD}}$, over six relaxed SD verifiers on one safety, two security, and two utility benchmarks. Figure~\ref{fig:step_l_curves_jailbreak} further reports the safety PRR on the Jailbreak benchmark as a function of the acceptance rate. Table~\ref{tab:agentdojo_results} shows the security results of evaluating \sys on AgentDojo. For each verifier and model family, we compare standard lossless SD, the vanilla relaxed verifier, and the corresponding verifier augmented with \sys.

\paragraph{\sys Mitigates Safety and Security Degradation}
Across most verifier-model combinations, \sys yields positive gains. For Qwen3-0.6B/8B, the average $\Delta_\mathrm{PRR}$ reaches $+0.270$ on Jailbreak and $+0.114$ on Prompt Injection. The largest single improvement is observed for SpecCascade, where \sys improves PRR by up to $+0.598$. Figure~\ref{fig:step_l_curves_jailbreak} further shows that vanilla relaxed verifiers lose PRR rapidly as the acceptance rate increases, whereas \sys keeps PRR close to that of the standard verifier. This confirms that correcting early accepted positions recovers most of the lost safety/security performance while retaining high acceptance rates.

\paragraph{\sys Preserves Utility Performance}
Unlike the substantial gains on safety/security benchmarks, \sys has little effect on utility benchmarks. For Qwen3-0.6B/8B, the average $\Delta_\mathrm{PRR}$ is $-0.008$ on HumanEval and $+0.046$ on GSM8K. 
These small changes show that early-position correction does not materially hurt task utility, consistent with our finding that utility-oriented tasks are less sensitive to the deviations corrected by \sys. In some cases, corrected prefixes even improve utility, such as the $+0.196$ PRR gain of SpecCascade on GSM8K.

\paragraph{\sys Introduces Negligible Efficiency Overhead}
\sys has only a minor impact on decoding efficiency, with $\Delta_{\mathrm{SPD}}$ remaining within $\pm 0.02$ across most settings. Stricter early verification may reject more draft tokens, but the corrected prefixes can improve later acceptance rates, making the overall speed impact nearly neutral. In some cases, \sys is even faster than the vanilla verifier. For example, SpecCascade gains $+0.083$ in speedup on Jailbreak with Qwen3-0.6B/8B, with an average gain of $+0.026$ in this setting.

\subsection{Effect of the Correction Length $L$}

Table~\ref{tab:eta_length} reports the effect of varying the step correction length $L \in \{1,2,4,8\}$ on $\Delta_\mathrm{PRR}$ and $\Delta_\mathrm{SPD}$. The two model families exhibit different sensitivities to $L$. For Qwen3-0.6B/8B, the safety gain on Jailbreak is already close to saturation at $L=1$, with the average $\Delta_\mathrm{PRR}$ increasing only slightly from $+0.270$ at $L=1$ to $+0.285$ at $L=8$. In contrast, Llama3-1B/8B benefits substantially from a longer correction window. Its average $\Delta_\mathrm{PRR}$ on Jailbreak increases from $+0.042$ at $L=1$ to $+0.321$ at $L=4$, and then plateaus at $+0.325$ when $L=8$.

Figure~\ref{fig:step_l_curves_jailbreak} shows the same trend across acceptance rates. For Qwen3, all Step curves separate clearly from the vanilla relaxed verifier even with $L=1$, and larger $L$ brings only marginal additional improvement. For Llama3, however, the Step curves move away from the vanilla verifier progressively as $L$ increases, indicating that correcting only the first accepted position is insufficient for the Llama3 family.

This difference is consistent with the distributional-deviation measurements in Figure~\ref{fig:tv_deviation}. For Qwen3, the per-position deviation $\mathrm{TV}(\tilde{q}_t,q_t)$ drops rapidly after the first few positions, suggesting that most harmful deviation is concentrated at the beginning of the generated response. Therefore, correcting the first accepted position already restores most of the target-consistent behavior. For Llama3, the deviation decays more slowly across positions, so relaxed verification can continue to steer the committed prefix away from the target model several tokens after the first acceptance. A longer correction window is thus needed to prevent this drift.

Increasing $L$ also introduces a mild efficiency cost. On Jailbreak, the average $\Delta_{\mathrm{SPD}}$ decreases from $+0.026$ at $L=1$ to $-0.047$ at $L=8$ for Qwen3, and from $+0.019$ to $-0.052$ for Llama3. Overall, $L$ provides a controllable trade-off between safety recovery and decoding speed: $L=1$ is sufficient for Qwen3, whereas $L=4$ recovers nearly all safety gains for Llama3 with only a small average speed change of $-0.008$.

\paragraph{How to Select $L$}
Table~\ref{tab:eta_length} suggests $L \leq 8$ as an empirical choice across tasks. Conceptually, the choice of $L$ should be related to the depth of safety alignment in the model~\cite{qi2025safety}. When $L$ matches the effective depth of safety alignment, the system can ensure that generation begins from a distribution consistent with the model's safety-aligned behavior. A more robust approach is to perform an ablation study over $L$ on a small calibration dataset and obtain curves similar to those in Figure~\ref{fig:eta_mitigation}, characterizing the trade-offs among Security/Safety-PRR, Utility-PRR, and SPD. One can then select $L$ according to the requirements of the deployment setting. When no representative calibration dataset is available, using a relatively large $L$ is a safer choice, since increasing $L$ has only a limited impact on the overall decoding speedup and is generally worthwhile given the corresponding improvement in system-level performance.

\subsection{Effect of Correction Schedule $\{\eta_j\}_{j\ge 1}$}

\begin{table}[t]
  \caption{
  \textbf{Impact of the correction schedule $\eta$ on \sys.} 
  S/L/P denote the STEP/LINEAR/POWER-LAW schedules, respectively. All three schedules effectively improve PRR. In practice, STEP avoids complex distribution-mixing operations, giving it an advantage in preserving SPD.
  }
  \vspace{-0.1in}
  \footnotesize
  \label{tab:eta_schedule}
  \begin{center}
    \setlength{\tabcolsep}{1pt}
    \begin{tabularx}{\columnwidth}{rc*{7}{>{\centering\arraybackslash}X}}
      \toprule
      \multirow{4}{*}{$\eta$ Schedule}
        & \multicolumn{4}{c}{Jailbreak-SD}
        & \multicolumn{4}{c}{HumanEval} \\
      \cmidrule(lr){2-5}\cmidrule(lr){6-9}
        & \multicolumn{2}{c}{Qwen3}
        & \multicolumn{2}{c}{Llama3}
        & \multicolumn{2}{c}{Qwen3}
        & \multicolumn{2}{c}{Llama3} \\
      \cmidrule(lr){2-3}\cmidrule(lr){4-5}\cmidrule(lr){6-7}\cmidrule(lr){8-9}
        & $\Delta_\mathrm{PRR}$ & $\Delta_\mathrm{SPD}$
        & $\Delta_\mathrm{PRR}$ & $\Delta_\mathrm{SPD}$
        & $\Delta_\mathrm{PRR}$ & $\Delta_\mathrm{SPD}$
        & $\Delta_\mathrm{PRR}$ & $\Delta_\mathrm{SPD}$ \\
      \midrule
      S ($L{=}1$)              & +.337 & +.050 & +.045 & +.003 & -.013 & +.001 & +.025 & -.002 \\
      S ($L{=}4$)              & +.349 & +.023 & +.292 & -.001 & -.004 & -.006 & +.035 & -.004 \\
      \midrule
      L ($L{=}1$)              & +.158 & -.039 & +.019 & -.010 & -.008 & -.012 & +.027 & -.010 \\
      L ($L{=}8$)              & +.313 & -.111 & +.207 & -.110 & -.009 & -.014 & +.037 & -.017 \\
      \midrule
      P ($\gamma{=}2, L{=}1$)   & +.248 & -.065 & +.027 & -.015 & +.005 & -.014 & +.026 & -.009 \\
      P ($\gamma{=}2, L{=}8$)   & +.349 & -.148 & +.266 & -.119 & -.015 & -.020 & +.033 & -.014 \\
      P ($\gamma{=}4, L{=}1$)   & +.313 & -.111 & +.036 & -.019 & -.007 & -.021 & +.034 & -.011 \\
      P ($\gamma{=}4, L{=}8$)   & +.347 & -.158 & +.291 & -.186 & +.012 & -.030 & +.038 & -.015 \\
      \bottomrule
    \end{tabularx}
  \end{center}
  \vspace{-0.15in}
\end{table}

The experiments above use the step schedule, which applies full correction to the first $L$ positions and then switches directly to the relaxed verifier. A natural question is whether a smoother transition would be preferable. Table~\ref{tab:eta_schedule} studies this on LossySD by replacing the step schedule with linear and power-law schedules, varying the transition length $L \in \{1, 8\}$ and the curvature $\gamma \in \{2, 4\}$, and comparing against the step schedule at $L \in \{1, 4\}$, under the same paired protocol on Jailbreak and HumanEval. The smooth schedules reproduce the qualitative behavior of the step schedule: security gains increase as the transition becomes longer and more delayed, reaching $+0.347$ on Qwen3 and $+0.291$ on Llama3 for the power-law schedule with $L=8$ and $\gamma=4$, while $\Delta_\mathrm{PRR}$ on HumanEval stays flat across all configurations.

However, the step schedule provides a better security--efficiency trade-off. On Qwen3, the step schedule with $L=1$ matches the best power-law gain ($+0.337$ versus $+0.349$) while being faster ($+0.050$ versus $-0.148$ in speed); on Llama3, the step schedule with $L=4$ matches it ($+0.292$ versus $+0.291$) with a much smaller speed cost ($-0.001$ versus $-0.186$). Intermediate correction levels are less efficient in our setting, as they extend correction beyond the most security-sensitive positions; the step schedule is also simpler to implement, since each position uses either standard or relaxed verification without materializing vocabulary-level mixtures. We therefore adopt the step schedule as the default configuration of \sys.
\begin{table}[t]
\centering
\footnotesize
\caption{
\textbf{Evaluation on batched inference with larger target models.}
We evaluate \sys with larger target models and report the average PRR and SPD across four datasets and six relaxed verifiers. \sys improves PRR while preserving SPD.
}
\label{tab:ablation_target_model}
\setlength{\tabcolsep}{3pt}
\begin{tabular}{llcccc}
\toprule
\multirow{3}{*}{Benchmark} & \multirow{3}{*}{\sys} & \multicolumn{2}{c}{Qwen3-0.6B/32B} & \multicolumn{2}{c}{Llama3-1B/70B} \\
\cmidrule(lr){3-4} \cmidrule(lr){5-6}
& & PRR & SPD & PRR & SPD \\
\midrule
\multirow{2}{*}{Jailbreak-SD} & w/o & 0.727 & 1.985 & 0.571 & 3.738 \\
& w/ & 0.973 & 1.752 & 0.690 & 3.947 \\
\midrule
\multirow{2}{*}{OpenPromptInjection} & w/o    & 0.652 & 1.853 & 0.512 & 3.153 \\
& w/ & 0.855 & 1.882 & 0.956 & 3.161 \\
\midrule
\multirow{2}{*}{HumanEval} & w/o    & 0.653 & 2.392 & 0.500 & 3.943 \\
& w/ & 0.627 & 2.538 & 0.667 & 3.874 \\
\midrule
\multirow{2}{*}{GSM8K} & w/o    & 0.820 & 2.637 & 0.791 & 3.785 \\
& w/ & 0.843 & 2.497 & 0.791 & 3.737 \\
\bottomrule
\end{tabular}
  \vspace{-0.1in}
\end{table}

\subsection{Ablation Study}

We further examine whether the gains of \sys remain stable under different deployment and decoding configurations, varying the target-model scale, the decoding temperature $\tau$, and the speculative length $k$. The $\tau$ and $k$ ablations use the same paired protocol on the two security benchmarks with Qwen3-0.6B/8B, averaged over the six relaxed verifiers and their hyperparameter sweeps.

\paragraph{Batched Inference with Larger Target Models}
Speculative decoding is often used in practice to accelerate larger models. In Table~\ref{tab:ablation_target_model}, we use Qwen3-32B and Llama3.3-70B-Instruct as target models and evaluate the improvements in PRR and SPD with and without \sys. We find that speculative decoding generally provides more pronounced speedups for larger target models. Meanwhile, \sys effectively improves the safety/security performance while preserving the acceleration benefits.

\paragraph{Decoding Temperature $\tau$}
Higher temperatures flatten the draft and target distributions and make the security-critical opening tokens more stochastic, which could plausibly affect \sys. Table~\ref{tab:ablation_temperature} varies $\tau \in \{0, 0.3, 0.7, 1.0\}$ with $k=5$ fixed. The security gain is largely insensitive to temperature: $\Delta_\mathrm{PRR}$ stays within $+0.238$ to $+0.274$ on Jailbreak and remains consistently positive on Prompt Injection ($+0.114$ to $+0.200$), while $\Delta_{\mathrm{SPD}}$ stays within $\pm0.026$. The benefit of \sys is therefore not an artifact of greedy decoding and carries over to stochastic sampling.

\begin{table}[t]
\centering
\setlength{\tabcolsep}{14pt}
\caption{
\textbf{Impact of the decoding temperature $\tau$ on \sys.}
Although different temperatures introduce generation uncertainty and PRR/SPD fluctuations, \sys remains effective across all tested temperatures.
}
  \vspace{-0.05in}
\label{tab:ablation_temperature}
\footnotesize
\begin{tabular}{lcccc}
\toprule
\multirow{3}{*}{$\tau$} & \multicolumn{2}{c}{Jailbreak-SD} & \multicolumn{2}{c}{OpenPromptInjection} \\
\cmidrule(lr){2-3}\cmidrule(lr){4-5}
& $\Delta_\mathrm{PRR}$ & $\Delta_\mathrm{SPD}$ & $\Delta_\mathrm{PRR}$ & $\Delta_\mathrm{SPD}$ \\
\midrule
0   & +.270 & +.026 & +.114 & -.014 \\
0.3 & +.269 & +.020 & +.155 & -.015 \\
0.7 & +.238 & +.014 & +.200 & -.017 \\
1.0 & +.274 & +.005 & +.168 & -.018 \\
\bottomrule
\end{tabular}
\end{table}

\begin{table}[t]
\centering
\footnotesize
\setlength{\tabcolsep}{14pt}
\caption{\textbf{Impact of the speculative length $k$ on \sys.}
Different speculative lengths affect the generation process and acceptance rate, while \sys maintains stable safety/security recovery across all speculative lengths.
}

  \vspace{-0.05in}
\label{tab:ablation_spec_length}
\begin{tabular}{l cc cc}
\toprule
\multirow{3}{*}{$k$} & \multicolumn{2}{c}{Jailbreak-SD} & \multicolumn{2}{c}{OpenPromptInjection} \\
\cmidrule(lr){2-3}\cmidrule(lr){4-5}
& $\Delta_\mathrm{PRR}$ & $\Delta_\mathrm{SPD}$ & $\Delta_\mathrm{PRR}$ & $\Delta_\mathrm{SPD}$ \\
\midrule
3 & +.266 & +.032 & +.114 & -.007 \\
4 & +.261 & +.037 & +.113 & -.012 \\
5 & +.270 & +.026 & +.114 & -.014 \\
6 & +.268 & +.030 & +.113 & -.016 \\
7 & +.270 & +.024 & +.128 & -.016 \\
8 & +.277 & +.018 & +.118 & -.016 \\
\bottomrule
\end{tabular}
\vspace{-0.15in}
\end{table}

\paragraph{Speculative Length $k$}
A larger $k$ commits more draft tokens per verification round and could in principle amplify the draft model's influence. Table~\ref{tab:ablation_spec_length} varies $k \in \{3, \dots, 8\}$ with $\tau=0$ fixed and shows that this does not occur: $\Delta_\mathrm{PRR}$ fluctuates only between $+0.261$ and $+0.277$ on Jailbreak and between $+0.113$ and $+0.128$ on Prompt Injection, while $\Delta_{\mathrm{SPD}}$ remains small ($-0.016$ to $+0.037$). This stability follows by design: the correction schedule is defined over global token positions, so \sys always corrects the first $L$ positions of the response regardless of how the speculation rounds are segmented.

\section{Discussion}

\paragraph{Adaptive Attacks}
Since \sys corrects only the first $L$ positions, an attacker aware of the defense could try to push the security-critical decision beyond the correction window. For example, the attacker could first elicit a benign-looking opening longer than $L$ tokens and then steer the subsequent generation toward harmful content. 
We instantiate this attack with 200 prompt injections that first repeat the original question to delay malicious behavior beyond the correction window, using the Step schedule ($L=12$). Table~\ref{tab:adaptive_attack} shows that \sys increases AFR by 19 percentage points and PRR by 55.9--71.8 percentage points over the average of six lossy-SD methods, approaching target-model performance. For queries where the target succeeds but the draft fails, early correction creates a target-aligned prefix that increases subsequent draft-target disagreement, reducing draft-token acceptance beyond the correction window and preserving safety/security.

\begin{table}[t]
\centering
\footnotesize
\caption{
\textbf{Evaluation under adaptive attacks.}
We evaluate \sys against adaptive prompt injections that delay malicious behavior beyond the correction window. By correcting early generation, \sys steers the response toward the target-model trajectory and remains effective against these adaptive attacks.
}
  \vspace{-0.05in}
\label{tab:adaptive_attack}
\setlength{\tabcolsep}{10.5pt}
\begin{tabular}{lcccc}
\toprule
\multirow{3}{*}{Setting} & \multicolumn{2}{c}{Qwen3-0.6B/8B} & \multicolumn{2}{c}{Llama3-1B/8B} \\
\cmidrule(lr){2-3} \cmidrule(lr){4-5}
& AFR$\uparrow$ & PRR$\uparrow$ & AFR$\uparrow$ & PRR$\uparrow$ \\
\midrule
Target   & .625 & --   & .580 & --   \\
Draft    & .285 & --   & .315 & --   \\
\midrule
w/o \sys & .415 & .382 & .360 & .168 \\
w/ \sys  & .605 & .941 & .550 & .886 \\
\midrule
$\Delta$ & +.190 & +.559 & +.190 & +.718 \\
\bottomrule
\end{tabular}
\end{table}

\begin{table}[t]
\centering
\setlength{\tabcolsep}{9pt}
\caption{\textbf{Runtime overhead of \sys.}
We report relative inference-latency changes in our vLLM implementation, in units of $10^{-3}\%$. Positive values indicate overhead and negative values indicate speedup; overhead is negligible.}
  \vspace{-0.05in}
\label{tab:runtime_overhead}
\footnotesize
\begin{tabular}{lcccc}
\toprule
& \multicolumn{2}{c}{Qwen3} & \multicolumn{2}{c}{Llama3} \\
\cmidrule(lr){2-3}\cmidrule(lr){4-5}
Method & 0.6B/8B & 0.6B/32B & 1B/8B & 1B/70B \\
\midrule
LossySD & +0.001 & +0.001 & +0.001 & +0.001 \\
BiLD & +0.001 & +0.001 & +0.001 & +0.001 \\
SpecCascade & -2.330 & -0.583 & -1.830 & -0.227 \\
FSD & +0.775 & +0.194 & +0.611 & +0.076 \\
FLy & +0.001 & +0.001 & +0.001 & +0.001 \\
MARS & +0.001 & +0.001 & +0.001 & +0.001 \\
\bottomrule
\end{tabular}
  \vspace{-0.15in}
\end{table}

\paragraph{Runtime Overhead in Real-World Systems}
Since \sys performs operations such as distribution fusion over the first $L$ tokens, it may introduce additional implementation overhead. We implement a lightweight \sys layer on vLLM and measure the time required to generate a fixed number of tokens across all datasets. Table~\ref{tab:runtime_overhead} shows that the additional overhead remains below $0.001\%$ because \sys reuses speculative-decoding intermediates and requires negligible computation relative to the target-model forward pass. For SpecCascade, \sys slightly improves speed by eliminating verifier-side recovery softmax and sampling. FSD incurs the largest positive overhead because it recomputes the full recovery distribution $R$, but the overhead remains negligible.

\paragraph{Early-Token Criticality for Safety and Security}
Prior studies~\cite{qi2025safety,zhang2024large} show that the first few generated tokens strongly influence the safety of a response. These studies examine intrinsic early-token sensitivity within a single model. We instead study system-level safety and security under lossy speculative decoding, where target-model decoding defines the performance upper bound, but relaxed verification can introduce early deviations from its generation trajectory. We show that these system-induced deviations degrade safety and security much more than utility, revealing a security--utility asymmetry not studied in prior work. This distinction motivates \sys, which restores stricter target verification at early positions.

\section{Conclusion}

In this paper, we presented the \emph{first} systematic study of the security implications of lossy speculative decoding. We revealed a clear safety/security--utility asymmetry, where relaxed verification improves efficiency but increases jailbreak and prompt injection risks much faster than it degrades utility. Our analysis attributed this asymmetry to early-position draft tokens, whose deviations from the target distribution shape the initial direction of the response. Based on this insight, we proposed \sys, a position-aware method that strengthens early verification while keeping efficient verification elsewhere. More broadly, our results show that inference optimizations can silently weaken model security, highlighting the need to include security evaluation in their utility and efficiency assessment.


\ifCLASSOPTIONcompsoc
  \section*{Acknowledgments}
\else
  \section*{Acknowledgment}
\fi

We thank the reviewers and shepherd for their constructive comments on our work. 
The Authors acknowledge the National Artificial Intelligence Research Resource (NAIRR) Pilot for contributing to this research result.
Neil Gong was supported in part by NSF under grant No. 2450935, 2112562, and  2625580.



%

\bibliographystyle{IEEEtran}
\bibliography{references}

\appendices

\section{Proofs of Theoretical Results}
\label{app:proofs}

\subsection{Proof of Theorem~\ref{thm:one_step_output_distribution}}
\label{app:proof_one_step_output_distribution}

\begin{proof}
Let $A_t$ denote the event that the draft token at step $t$ is accepted. Since $\tilde{x}_t\sim p_t$ and the verifier accepts token $u$ with probability $\alpha_t(u)$, the overall acceptance rate is
\begin{align}
\Pr_S(A_t\mid x_{<t})
= \sum_{u\in\mathcal{V}} p_t(u) \alpha_t(u)
= \mathbb{E}_{u\sim p_t}[\alpha_t(u)].
\end{align}
We expand $\tilde{q}_t(x)$ by conditioning on whether the draft token is accepted or rejected:
\begin{align}
\tilde{q}_t(x)
&= \Pr_S(x_t=x, A_t\mid x_{<t})\notag \\
&+ \Pr_S(A_t^c\mid x_{<t})\Pr_S(x_t=x\mid A_t^c,x_{<t}).
\end{align}
For the first term, using $\Pr_S(\tilde{x}_t=x\mid x_{<t})=p_t(x)$ and $\Pr_S(A_t\mid\tilde{x}_t=x,x_{<t})=\alpha_t(x)$:
\begin{align}
\Pr_S(x_t=x,A_t\mid x_{<t}) = p_t(x)\alpha_t(x).
\end{align}
For the second term, rejection triggers sampling from $R_t$, so $\Pr_S(x_t=x\mid A_t^c,x_{<t})=R_t(x)$. This gives
\begin{align}
\tilde{q}_t(x)
= p_t(x)\alpha_t(x)
+ \bigl(1-\mathbb{E}_{u\sim p_t}[\alpha_t(u)]\bigr)R_t(x),
\end{align}
which completes the proof.
\end{proof}

\subsection{Proof of Corollary~\ref{cor:lossless_sd}}
\label{app:proof_lossless_sd}

\begin{proof}
Applying Theorem~\ref{thm:one_step_output_distribution} with $\alpha_t = \alpha_t^{\mathrm{sd}}$ and $R_t = R_t^{\mathrm{sd}}$, the acceptance term evaluates to
\begin{align}
p_t(x)\alpha_t^{\mathrm{sd}}(x)
&= p_t(x)\min\left(1,\frac{q_t(x)}{p_t(x)}\right) \\
&= \min\bigl(p_t(x),q_t(x)\bigr).
\end{align}
Since $\sum_{u}(p_t(u)-q_t(u))=0$, the positive and negative parts have equal total mass, so the rejection probability satisfies
\begin{align}
1 - \mathbb{E}_{u\sim p_t}[\alpha_t^{\mathrm{sd}}(u)]
&= \sum_{u\in\mathcal{V}}(p_t(u)-q_t(u))_+ \\
&= \sum_{u\in\mathcal{V}}(q_t(u)-p_t(u))_+.
\end{align}
Substituting both into~\eqref{eq:induced_output_distribution} gives
\begin{align}
\tilde{q}_t^{\mathrm{sd}}(x)
&= q_t(x),
\end{align}
which completes the proof.
\end{proof}

\subsection{Proof of Theorem~\ref{thm:distributional_deviation_decomposition}}
\label{app:proof_distributional_deviation_decomposition}

\begin{proof}
Let $\bar\alpha_t\triangleq\mathbb{E}_{u\sim p_t}[\alpha_t(u)]$ and $\bar\alpha_t^{\mathrm{sd}}\triangleq\mathbb{E}_{u\sim p_t}[\alpha_t^{\mathrm{sd}}(u)]$. Substituting $\tilde{q}_t$ from Theorem~\ref{thm:one_step_output_distribution} and $q_t$ from Corollary~\ref{cor:lossless_sd},
\begin{align}
\mathrm{TV}(\tilde{q}_t,q_t)
&= \frac{1}{2}\sum_{x\in\mathcal{V}}\Bigl|p_t(x)\alpha_t(x)+(1-\bar\alpha_t)R_t(x) \notag\\
&\quad-p_t(x)\alpha_t^{\mathrm{sd}}(x)-(1-\bar\alpha_t^{\mathrm{sd}})R_t^{\mathrm{sd}}(x)\Bigr|.
\end{align}
Using the definition $\delta_t(x)=\alpha_t(x)-\alpha_t^{\mathrm{sd}}(x)$, which gives $\bar\alpha_t-\bar\alpha_t^{\mathrm{sd}}=\mathbb{E}_{p_t}[\delta_t]$, we collect terms:
\begin{align}
\mathrm{TV}(\tilde{q}_t,q_t)
&= \frac{1}{2}\sum_{x\in\mathcal{V}}\Bigl|p_t(x)\delta_t(x)-\mathbb{E}_{p_t}[\delta_t]R_t^{\mathrm{sd}}(x) \notag\\
&\quad+(1-\bar\alpha_t)\bigl(R_t(x)-R_t^{\mathrm{sd}}(x)\bigr)\Bigr|.
\end{align}
Applying the triangle inequality separates the contribution of the acceptance from that of the recovery, which yields
\begin{align}
\mathrm{TV}(\tilde{q}_t,q_t)
&\le \frac{1}{2}\sum_{x\in\mathcal{V}}\bigl|p_t(x)\delta_t(x)-\mathbb{E}_{p_t}[\delta_t]R_t^{\mathrm{sd}}(x)\bigr| \notag\\
&\quad+(1-\bar\alpha_t)\mathrm{TV}(R_t,R_t^{\mathrm{sd}}).
\end{align}
For the first term, we apply the triangle inequality once more and use the condition $\sum_{x}R_t^{\mathrm{sd}}(x)=1$. This yields
\begin{align}
&\frac{1}{2}\sum_{x\in\mathcal{V}}\bigl|p_t(x)\delta_t(x)-\mathbb{E}_{p_t}[\delta_t]R_t^{\mathrm{sd}}(x)\bigr| \notag\\
&\le\frac{1}{2}\mathbb{E}_{p_t}[|\delta_t|]
+\frac{1}{2}\bigl|\mathbb{E}_{p_t}[\delta_t]\bigr| \le \mathbb{E}_{p_t}[|\delta_t|],
\end{align}
where the last step uses $|\mathbb{E}_{p_t}[\delta_t]|\le\mathbb{E}_{p_t}[|\delta_t|]$. Combining both bounds gives the stated inequality.
\end{proof}

\subsection{Proof of Theorem~\ref{thm:bd_bound}}
\label{app:proof_bd_bound}

\begin{proof}
Let $\mathrm{TV}_t(x\mid x_{<t})\triangleq\tilde{q}_t(x\mid x_{<t})-q_t(x\mid x_{<t})$ be the signed pointwise deviation. Define hybrids $H^{(s)}$ using $\tilde{q}_t$ for $t\le s$ and $q_t$ for $t>s$, so $H^{(0)}=Q$ and $H^{(T)}=\tilde{Q}$. Telescoping over the hybrid chain yields:
\begin{align}
\mathbb{E}_{\tilde{Q}}[M]-\mathbb{E}_{Q}[M]
=\sum_{t=1}^{T}\bigl(\mathbb{E}_{H^{(t)}}[M]-\mathbb{E}_{H^{(t-1)}}[M]\bigr).
\end{align}
Since $H^{(t)}$ and $H^{(t-1)}$ differ only at position $t$ and share the same prefix marginal $\tilde{Q}_{<t}$, the $t$-th term equals $\Delta_t$, where
\begin{align}
\Delta_t \triangleq \mathbb{E}_{x_{<t}\sim\tilde{Q}_{<t}}\Bigl[
\sum_{x\in\mathcal{V}}\mathrm{TV}_t(x\mid x_{<t})\mu_t(x;x_{<t})\Bigr].
\end{align}
Hence $\mathbb{E}_{\tilde{Q}}[M]-\mathbb{E}_{Q}[M]=\sum_{t}\Delta_t$. For each fixed $x_{<t}$, the TV-oscillation inequality gives $\bigl|\sum_x\mathrm{TV}_t(x\mid x_{<t})\mu_t(x;x_{<t})\bigr|\le\mathrm{TV}_t(x_{<t})\cdot\mathrm{osc}\bigl(\mu_t(\cdot;x_{<t})\bigr)$. Taking expectations over $x_{<t}\sim\tilde{Q}_{<t}$,
\begin{align}
|\Delta_t| \le \mathbb{E}_{x_{<t}\sim\tilde{Q}_{<t}}\bigl[
\mathrm{TV}_t(x_{<t})\cdot\mathrm{osc}\bigl(\mu_t(\cdot;x_{<t})\bigr)\bigr].
\end{align}
Summing over $t$ and applying the triangle inequality gives $|\mathbb{E}_{\tilde{Q}}[M]-\mathbb{E}_Q[M]|\le\Gamma$. The uniform bound~\eqref{eq:bd_bound_uniform} then follows by bounding $\mathrm{TV}_t(x_{<t})\le\mathrm{TV}_t$ and $\mathrm{osc}(\mu_t(\cdot;x_{<t}))\le\mathrm{osc}(\mu_t)$ pointwise for each prefix $x_{<t}$.
\end{proof}

\subsection{Proof of Proposition~\ref{prop:eta_mixing_deviation}}
\label{app:proof_eta_mixing_deviation}

\begin{proof}
Expanding the definitions of $\alpha_t^\eta$ and $R_t^\eta$ gives
\begin{align}
\tilde{q}_t^{\eta}&=(1-\eta_t)q_t+\eta_t\tilde{q}_t\notag\\
&+\eta_t(1-\eta_t)(\bar{\alpha}_t-\bar{\alpha}_t^{\mathrm{sd}})(R_t-R_t^{\mathrm{sd}}).
\end{align}
Therefore, by the triangle inequality,
\begin{align}
\mathrm{TV}(\tilde{q}_t^{\eta},q_t) &\le \eta_t\mathrm{TV}(\tilde{q}_t,q_t)\notag\\
&+\eta_t(1-\eta_t)|\bar{\alpha}_t-\bar{\alpha}_t^{\mathrm{sd}}|\mathrm{TV}(R_t,R_t^{\mathrm{sd}}).
\end{align}
Taking uniform upper bounds over prefixes and substituting this result into Theorem~\ref{thm:bd_bound} yields
\begin{align}
\left|\mathbb{E}_{\tilde{Q}^{\eta}}[M]\!-\!\mathbb{E}_{Q}[M]\right|\!\le\!\min\left\{\sum_{t=1}^{T}(\eta_t\mathrm{TV}_t+\epsilon_t)I_t,1\right\}.
\end{align}
\end{proof}

\section{Lossless Baseline}
\label{app:additional_results}

Table~\ref{tab:lossless_secure_comparison} compares the PRR and SPD of lossless SD, average lossy SD, and \sys across four tasks and two model pairs.

\begin{table}[h]
\centering
\caption{\textbf{Lossless baseline.}
PRR and SPD are measured relative to target decoding. Lossy and \sys (STEP, $L=1$) are averaged over six lossy-SD methods.}
\label{tab:lossless_secure_comparison}
\footnotesize
\setlength{\tabcolsep}{2pt}
\begin{tabular}{@{}lccc|ccc@{}}
\toprule
\multirow{2}{*}{Task}
& \multicolumn{3}{c|}{PRR$\uparrow$}
& \multicolumn{3}{c}{SPD$\uparrow$ ($\times$)} \\
\cmidrule(lr){2-4}\cmidrule(l){5-7}
& Lossless & Lossy & \sys
& Lossless & Lossy & \sys \\
\midrule
\multicolumn{7}{@{}l}{\emph{Qwen3-0.6B/8B}} \\
Jailbreak
& 1.019 & .782 & 1.052
& 1.049 & 1.543 & 1.569 \\
Prompt Inj.
& .915 & .700 & .814
& 1.057 & 1.439 & 1.425 \\
HumanEval
& .947 & .874 & .866
& 1.646 & 1.762 & 1.754 \\
GSM8K
& 1.000 & .932 & .978
& 1.807 & 1.818 & 1.811 \\
\midrule
\multicolumn{7}{@{}l}{\emph{Llama3-1B/8B}} \\
Jailbreak
& .909 & .682 & .724
& 1.326 & 1.736 & 1.754 \\
Prompt Inj.
& .983 & .755 & .854
& 1.364 & 1.444 & 1.430 \\
HumanEval
& .933 & .814 & .832
& 1.849 & 1.854 & 1.855 \\
GSM8K
& .976 & .915 & .930
& 1.764 & 1.853 & 1.852 \\
\bottomrule
\end{tabular}
\end{table}

\section{Ethics Considerations}
\label{app:ethics}

\paragraph{Stakeholder Analysis}
Consider three stakeholder groups:
\begin{icompact}
\item \emph{Model providers and infrastructure operators} who deploy LLM serving systems and develop serving frameworks (e.g., vLLM, TensorRT-LLM, and SGLang) that support lossy speculative decoding. Our findings reveal that relaxed verification, an efficiency optimization usually evaluated only on utility benchmarks, can substantially weaken jailbreak and prompt injection resilience. This can inform safer verification designs and the inclusion of security evaluation when tuning inference optimizations. The potential negative impact is that an adversary could preferentially target services that deploy aggressive relaxed verification.

\item \emph{End users} of LLM systems, whose safety depends on the target model's alignment being preserved at inference time. They may benefit from defenses that keep the system's behavior consistent with the safety-aligned target model, while unprotected lossy speculative decoding could expose them to more harmful responses or to unintended actions triggered by prompt injection in downstream applications.

\item \emph{Researchers and practitioners} studying LLM systems, who benefit from understanding that inference optimizations are security-relevant components. Our work characterizes the security--utility asymmetry and its origin in early-position draft tokens, and may motivate further defensive research, while this analysis could also guide attackers toward settings where existing defenses are weaker.
\end{icompact}

\paragraph{Mitigation}
We took steps to minimize harm to the identified stakeholders and promote defensive progress.
\begin{icompact}
\item \emph{For model providers and infrastructure operators, a practical defense}: Alongside the measurement, we propose \sys, which restores safety and security performance close to that of lossless speculative decoding while preserving efficiency and utility.

\item \emph{For end users, no real-world harm}: All experiments used public benchmarks (WildJailbreak, JailbreakBench, HarmBench, AdvBench, OpenPromptInjection, AgentDojo, HumanEval, and GSM8K) and open-weight models (Qwen3 and Llama3 families) deployed locally in controlled offline environments. We did not use private user data, interact with real users, or issue attack queries to production services. GPT-4o-mini is used only as an LLM-as-a-judge to classify model outputs, and we do not include harmful model generations in the paper.

\item \emph{For researchers and practitioners, no new attack technique}: All jailbreak and prompt injection inputs are drawn from existing public datasets, and our adaptive attack only applies a simple prefix-delaying transformation to them. We release the code of \sys and the Jailbreak-SD benchmark, which is constructed solely from existing public datasets, to support reproducible defensive research rather than to provide exploit guidance.
\end{icompact}

\paragraph{Recommended Future Safeguards}
Real-world systems should incorporate safeguards addressing risks to both infrastructure operators and end users. Specifically, prior to deployment, systems utilizing relaxed verification must be evaluated against safety and security benchmarks alongside standard utility metrics, with acceptance rates capped at levels where any security degradation remains strictly tolerable.

\paragraph{Justification for Research}
Lossy speculative decoding improves LLM inference efficiency and is already supported by mainstream serving frameworks, but its security implications remain largely unexplored. Identifying these risks and providing a practical defense is essential to ensure that efficiency gains do not silently compromise the safety and security of LLM systems.

\newpage

\section{Meta-Review}

The following meta-review was prepared by the program committee for the 2027
IEEE Symposium on Security and Privacy (S\&P) as part of the review process as
detailed in the call for papers.

\subsection{Summary}
The paper studies lossy speculative decoding and identifies a security–utility asymmetry. The resilience of the system against jailbreak and prompt-injection degrades substantially earlier and faster than standard utility metrics. This is attributed to the disproportionate impact (and susceptibility) of the early tokens. To address this, the work proposes a defence -- SecureSD -- the lowers attack effectiveness.

\subsection{Scientific Contributions}
\begin{itemize}
\item Identifies an Impactful Vulnerability
\item Provides a Valuable Step Forward in an Established Field
\end{itemize}

\subsection{Reasons for Acceptance}
\begin{enumerate}
\item The work identifies a new threat by observing that under lossy speculative decoding, the susceptibility of the attacks degrades faster than model utility.
\item The work introduces the first defence to address this -- SecureSD -- that relies on interpolating between the draft and judge models to get the output distribution closer to that of lossless speculative decoding. It leverages the known insight that early tokens have higher impact on attack effectiveness.
\end{enumerate}

\subsection{Noteworthy Concerns} 
\begin{enumerate} 
\item The security guarantees of SecureSD and the evaluation of the adaptive attacks are a good first step, but overall limited.
\end{enumerate}

\section{Response to the Meta-Review} 

We thank the reviewers and would like to clarify \sys's scope. \sys mitigates the additional security degradation introduced by lossy speculative decoding, with target-only decoding as its security reference. It does not claim to eliminate the target model’s existing vulnerabilities or provide universal adversarial robustness.

Our adaptive evaluation directly targets \sys's design by delaying malicious generation beyond its early-token correction window. The results in Table 12 show that \sys retains mitigation benefits under this defense-aware attack. We acknowledge that evaluating more diverse adaptive strategies is an important direction for future work.

\end{document}